\documentclass[aps,pra,superscriptaddress,amsmath,amssymb,twocolumn,footinbib]{revtex4-1}

\usepackage{graphicx}
\usepackage[dvipsnames,svgnames]{xcolor}
\usepackage{bm,bbm}%
\usepackage{braket}
\usepackage{amsthm,amsmath,amssymb,amsbsy,arev}
\usepackage{enumitem}
\usepackage{subcaption,caption}
\usepackage{booktabs,multirow}
\usepackage{mathtools,mathrsfs,bbding}
\usepackage[percent]{overpic}
\usepackage{microtype}
\usepackage{array,adjustbox,afterpage}
\usepackage[T1]{fontenc}
\usepackage[utf8]{inputenc}

\newtheoremstyle{red}{}{}{\itshape}{}{\color{red!80!black}\bfseries}{.}{ }{}

\definecolor{darkred}{rgb}{0.57,0,0.12}
\usepackage[colorlinks=true, linkcolor=MidnightBlue, urlcolor=darkred!75!black, citecolor=MidnightBlue]{hyperref}

\newcommand{\nc}{\newcommand}
\usepackage{cleveref}
\usepackage[most,breakable]{tcolorbox}

\AtBeginDocument{
\heavyrulewidth=.08em
\lightrulewidth=.05em
\cmidrulewidth=.03em
\belowrulesep=.65ex
\belowbottomsep=0pt
\aboverulesep=.4ex
\abovetopsep=0pt
\cmidrulesep=\doublerulesep
\cmidrulekern=.5em
\defaultaddspace=.5em
}

\makeatletter
\renewcommand{\p@subsection}{}
\renewcommand{\p@subsubsection}{}
\def\l@subsubsection#1#2{}
 \def\@hangfrom@section#1#2#3{\@hangfrom{#1#2}#3}%
 \def\@hangfroms@section#1#2{#1{#2}}%
 \def\@alph#1{\ifcase#1\or a\or b\or c\or d\else\@arabic{#1}\fi}
\makeatother

\usepackage{pxfonts}
\usepackage{etoolbox}
\hypersetup{linktoc=all}

\nc{\ketbra}[2]{\ket{#1}\!\bra{#2}}

\DeclareMathOperator{\Tr}{Tr}

\DeclareMathOperator{\NN}{\mathbb{N}}
\DeclareMathOperator{\supp}{supp}

\DeclareMathOperator{\conv}{conv}
\DeclareMathOperator{\cl}{cl}

\DeclareMathOperator{\rank}{rank}

\renewcommand{\H}{\mathcal{H}}

\newcommand{\norm}[2]{\left\lVert#1\right\rVert_{\,#2}}
\newcommand{\proj}[1]{\ket{#1}\!\bra{#1}}

\newcommand{\lnorm}[2]{\left\lVert#1\right\rVert_{\ell_{#2}}}

\nc{\TT}[1]{T^{({#1})}_\S}
\nc{\TTJ}[1]{T^{({#1})}_\J}

\renewcommand{\S}{\mathcal{S}}

\newcommand{\T}{\mathcal{T}}

\newcommand{\V}{\mathcal{V}}

\newcommand{\M}{\mathcal{M}}
\newcommand{\N}{\mathcal{N}}
\newcommand{\I}{\mathcal{I}}
\newcommand{\Q}{\mathcal{Q}}

\newcommand{\mleq}{\preceq}
\newcommand{\mgeq}{\succeq}

\renewcommand{\*}{\textup{*}}
\newcommand{\<}{\left\langle}
\renewcommand{\>}{\right\rangle}

\renewcommand{\bar}{\;\rule{0pt}{9.5pt}\right|\;}

\newcommand{\lset}{\left\{\left.}
\newcommand{\rset}{\right\}}

\DeclareMathOperator*{\argmin}{arg\,min}

\newcommand{\RR}{\mathbb{R}}
\newcommand{\CC}{\mathbb{C}}

\newcommand{\HH}{\mathbb{H}}
\newcommand{\DD}{\mathbb{D}}

\newcommand{\ve}{\varepsilon}

\newcommand{\cbraket}[1]{\left|\braket{#1}\right|}

\newcommand{\id}{\mathbbm{1}}
\newcommand{\idc}{\text{\rm id}}

\renewcommand{\O}{\mathcal{O}}

\newcommand{\J}{\mathcal{J}}
\newcommand{\GG}[1]{G^{({#1})}}
\newcommand{\TTT}[1]{T^{({#1})}}

\newcommand{\mc}{{\text{\rm mc}}}

\newcommand{\toLOCC}{\xrightarrow{\LOCC}}
\newcommand{\SEPmcdd}{{\text{\rm SEP}_{\mc}\hspace{-1.7ex}\raisebox{0.2ex}{\*\*}}}

\newcommand{\LOCC}{\text{\rm LOCC}}

\newcommand{\OLOCC}{{1\text{\rm{}-LOCC}}}

\newcommand{\LOCCO}{\text{\rm LOCC}}

\newcommand{\PPTO}{{\text{\rm PPT}}}
\newcommand{\SEPO}{\text{\rm SEP}}

\newcommand{\SEPP}{\text{\rm SEPP}}
\nc{\PPTPPR}{\text{\rm PPTP}_+}
\nc{\PPTPR}{\text{\rm PPTP}}

\nc{\ppt}{\text{\rm\sffamily PPT}}
\nc{\pptp}{\text{\rm\sffamily PPT}_{+}}

\newcommand{\PPT}{\ppt}
\newcommand{\PPTP}{{\pptp}}
\newcommand{\PPTR}{{\PPT^{\hspace{0.1em}\prime}}}
\nc{\PPTRP}{{{\PPT^{\hspace{0.1em}\prime}_+}}}
\nc{\PPTRPPR}{{\text{\rm PPTP}^{\hspace{0.1em}\prime}_+}}
\newcommand{\SEP}{\text{\rm\sffamily SEP}}
\newcommand{\PPTPd}{{\text{\rm\sffamily PPT}_{+}\hspace{-0.85ex}\raisebox{0.2ex}{\*}}}

\newcommand{\PPTRc}{{\PPTR^\circ}}

\newcommand{\mnorm}[1]{\norm{#1}{[m]}}

\newcommand{\Qc}{\Q^\circ}
\newcommand{\Qcc}{\Q^{\circ\circ}}
\newcommand{\cc}{{\circ\circ}}
\nc{\wt}{\widetilde}

\nc{\logfloor}[1]{\left\lfloor {#1} \right\rfloor_{\log}}
\nc{\logceil}[1]{\left\lceil {#1} \right\rceil_{\log}}

\renewenvironment{boxed}[1]%
  {\expandafter\ifstrequal\expandafter{#1}{orange}{\begin{tcolorbox}[colback=orange!3,colframe=orange!15]}{\begin{tcolorbox}[colback=white,colframe=gray!10,breakable,enhanced]}}%
  {\end{tcolorbox}}

\newcommand{\FF}[1]{F_{\text{#1}}}

{\begin{boxed}{orange}\vspace{-\baselineskip}\begin{equation}}%
{\end{equation}\end{boxed}}

\newtheorem{theorem}{Theorem}
\newtheorem{proposition}[theorem]{Proposition}
\newtheorem{corollary}[theorem]{Corollary}
\newtheorem{definition}[theorem]{Definition}
\newtheorem{lemma}[theorem]{Lemma}

\theoremstyle{red}

\theoremstyle{definition}
\newtheorem*{remark}{Remark}

\let\oldproofname\proofname
\renewcommand{\proofname}{\rm\bf{\oldproofname}}

\patchcmd{\section}
  {\centering}
  {\raggedright\large}
  {}
  {}
\patchcmd{\subsection}
  {\centering}
  {\raggedright\large}
  {}
  {}
  \patchcmd{\subsubsection}
  {\centering}
  {\raggedright\normalsize\bfseries\textup}
  {}
  {}
\let\nc\newcommand
  \nc{\MIO}{{\text{\rm MIO}}}
\nc{\DIO}{{\text{\rm DIO}}}
\nc{\SIO}{{\text{\rm SIO}}}
\nc{\IO}{{\text{\rm IO}}}
\nc{\lsetr}{\left\{\,}
\nc{\rsetr}{\right.\right\}}
\nc{\barr}{\,\rule{0pt}{9.5pt}\left|\;}

\makeatletter
\def\l@f@section{%
 \addpenalty{\@secpenalty}%
 \addvspace{0.4em plus\p@}%
}%

\makeatother
\makeatletter
\patchcmd{\@ssect@ltx}
    {\addcontentsline{toc}{#1}{\protect\numberline{}#8}}
    {}
    {}
    {}
\makeatother

\begin{document}

\title{\Large One-shot entanglement distillation\\beyond local operations and classical communication}

\author{Bartosz Regula}
\email{bartosz.regula@gmail.com}
\affiliation{School of Physical and Mathematical Sciences, Nanyang Technological University, 637371, Singapore}
\affiliation{Complexity Institute, Nanyang Technological University, 637335, Singapore}

\author{Kun Fang}
\email{kf383@cam.ac.uk}
\affiliation{Department of Applied Mathematics and Theoretical Physics, University of Cambridge, Cambridge, CB3 0WA, UK}

\author{Xin Wang}
\email{wangxinfelix@gmail.com}
\affiliation{Joint Center for Quantum Information and Computer Science, University of Maryland, College Park, Maryland 20742, USA}

\author{Mile Gu}
\email{mgu@quantumcomplexity.org}
\affiliation{School of Physical and Mathematical Sciences, Nanyang Technological University, 637371, Singapore}
\affiliation{Complexity Institute, Nanyang Technological University, 637335, Singapore}
\affiliation{Centre for Quantum Technologies, National University of Singapore, 3 Science Drive 2, 117543, Singapore}

\begin{abstract}%
We study the task of entanglement distillation in the one-shot setting under different classes of quantum operations which extend the set of local operations and classical communication (LOCC). Establishing a general formalism which allows for a straightforward comparison of their exact achievable performance, we relate the fidelity of distillation under these classes of operations with a family of entanglement monotones, and the rates of distillation with a class of smoothed entropic quantities based on the hypothesis testing relative entropy. We then characterise exactly the one-shot distillable entanglement of several classes of quantum states and reveal many simplifications in their manipulation.\vspace{.5\baselineskip}\\
We show in particular that the $\varepsilon$-error one-shot distillable entanglement of any pure state is the same under all sets of operations ranging from one-way LOCC to separability-preserving operations or operations preserving the set of states with positive partial transpose, and can be computed exactly as a quadratically constrained linear program. We establish similar operational equivalences in the distillation of isotropic and maximally correlated states, reducing the computation of the relevant quantities to linear or semidefinite programs. We also show that all considered sets of operations achieve the same performance in environment-assisted entanglement distillation from any state.
\end{abstract}

\maketitle

\vspace{-\baselineskip}

\section{Introduction}
Quantum entanglement plays a fundamental role in quantum information processing by serving as a resource which underlies many important protocols such as quantum teleportation~\cite{bennett_1993} or superdense coding~\cite{bennett_1992} as well as quantum technological applications such as quantum repeaters and networks~\cite{briegel_1998,kimble_2008-1}. Many such schemes require the use of entanglement in the pure, maximal form of singlets --- the efficient conversion of entanglement into such form, dubbed \textit{entanglement distillation}~\cite{bennett_1996,bennett_1996-3}, is thus of vital importance, and the development of effective theoretical and practical methods to characterise entanglement distillation remains at the forefront of quantum information research~\cite{horodecki_2009}. First studied in the asymptotic regime under the assumption of being able to manipulate an unbounded number of
independent and identically distributed copies of a quantum system~\cite{bennett_1996,bennett_1996-3,rains_1999,horodecki_1998,vedral_1998,devetak_2005}, distillation later attracted a significant amount of research using the tools of non-asymptotic quantum information theory~\cite{lo_2001,morikoshi_2001,martin-delgado_2003,hayashi_2006-1,buscemi_2010-1,buscemi_2013,brandao_2011,datta_2015,fang_2019,rozpedek_2018}. The latter setting is of particular importance due to the physical limitations of near-term quantum technologies, preventing us from being able to manipulate large numbers of quantum systems effectively. In particular, to efficiently exploit entanglement in practical settings it is necessary to obtain a thorough understanding of \textit{one-shot} distillation of entanglement, which takes into account the realistic, non-asymptotic restrictions on state transformations and aims to understand how finite accuracy limits our ability to manipulate entanglement.

The characterisation of entanglement as a resource in practical settings is rooted in the so-called distant labs paradigm~\cite{horodecki_2009}, in which experimenters are free to perform any local operation within their own labs and communicate with each other classically, but any use of quantum communication has an associated resource cost since it requires the use of entanglement. This formalism led to the definition of local operations and classical communication (LOCC) as the set of allowed (``free'') operations, and the operational characterisation of entanglement distillation is concerned  precisely with delimiting the capabilities of LOCC in manipulating entanglement. However, the mathematical description of LOCC is known to have a highly complicated structure~\cite{chitambar_2014}, making many important questions in the resource theory of entanglement either very challenging or downright unanswerable. This motivated the investigation of several relaxations of the class LOCC~\cite{rains_1997,rains_1999-1,rains_2001,eggeling_2001,brandao_2010}, whose simplified description can provide accessible upper bounds on the capabilities of LOCC as well as establish the ultimate limitations on entanglement transformations. Understanding the properties of such relaxations and characterising their precise operational power can therefore shed light on the fundamental structure of entanglement as a resource. 

In this work, we develop a comprehensive framework for the study of one-shot entanglement distillation under several different classes of operations --- separable maps (SEP), separability-preserving maps (SEPP), positive partial transpose (PPT) maps, two types of positive partial transpose--preserving maps, as well as two types of maps based on the so-called Rains set --- many of which have been considered in the literature as a relaxation of LOCC in various contexts, but whose one-shot distillation capabilities in relation to other operations remained unknown. Such extensions are still bound by operationally motivated constraints (e.g., SEPP can never generate entanglement from an unentangled state, just as LOCC), but they can often be understood as allowing for additional resources to be used in entanglement manipulation (e.g, any PPT operation can be stochastically implemented by LOCC if one is additionally given access to a bound entangled state~\cite{cirac_2001}). We compare the performance of these sets of maps in distilling entanglement in the one-shot setting, establishing in particular a general formalism which allows us to describe the distillation under the different operations together in a unified framework. We make use of tools from convex analysis and convex optimisation to relate the rates of distillation with a family of entanglement monotones. By evaluating these monotones for all pure states, isotropic states, and maximally correlated states, we simplify the description of distillation in these cases and show that many of the relaxations coincide in their distillation power, facilitating an efficient quantification of fundamental entanglement properties and revealing many operational similarities in entanglement manipulation under different classes of channels.

Our work improves many earlier results in the characterisation of one-shot entanglement distillation~\cite{buscemi_2010-1,buscemi_2013,datta_2015,wilde_2017-2,fang_2019}, which relied on approximate bounds and were only exact asymptotically; crucially, our formalism allows for a precise description of distillation already at the one-shot level, providing an exact characterisation of the operational power of several classes of operations which extend LOCC and shedding light on the capabilities of LOCC themselves.

\subsection{Summary of results}

We begin our work in Sec.~\ref{sec:prelim} with a brief introduction to a family of entanglement monotones $\TTT{m}_\S$ which will play an important role in the later investigation of entanglement distillation. We characterise their properties and in particular show that the class of monotones generalises two known measures of entanglement --- the robustness of entanglement and a distance-based quantifier based on trace distance --- which will allow us to endow the measures with a direct operational meaning.

Our characterisation of entanglement distillation begins in Section~\ref{sec:dist} where we establish explicit general connections between the quantifiers $\TTT{m}_\S$, quantum hypothesis testing, and one-shot entanglement distillation through convex duality. The methods will form the foundations of the framework developed in this work.

We commence the explicit applications of our framework in Sec.~\ref{sec:dist_ppt} by quantifying the distillation capabilities of several classes of operations based on the set of PPT states, recovering previous results of~\cite{rains_2001,eggeling_2001} as well as describing new classes of operations in this context. The results additionally allow for an understanding of important asymptotic quantities, such as the regularised PPT relative entropy of entanglement or the Rains bound \cite{rains_1999-1,rains_2001}, not just as bounds for distillable entanglement but as quantities with a precise operational meaning of their own. This section serves also as an introduction to the formalism considered in the manuscript and showcases the generality of our methods.

In Section~\ref{sec:sepp}, we consider the class of separability-preserving operations~\cite{brandao_2010,brandao_2011}. By relating the achievable fidelity of distillation with the monotones $\TTT{m}_\S$ again, we establish an operational interpretation of the generalised robustness of entanglement in the context of distillation. Furthermore, we demonstrate a general operational equivalence in the distillation from pure states: all sets of operations, ranging from one-way LOCC to SEPP and PPT-preserving operations, achieve exactly the same performance in one-shot pure-state distillation. Although such an equivalence in the asymptotic regime was already known~\cite{bennett_1996-1,vedral_1998}, the correspondence already in the one-shot setting is remarkable, considering that the one-shot manipulation power of the larger sets of operations is generally much greater than that of LOCC. The results allow us to explicitly relate the fidelity of distillation of any pure state with an analytically computable norm of its Schmidt coefficients and express the computation of the $\ve$-error one-shot distillable entanglement of a pure state as a convex quadratically-constrained linear program.

We continue in Section~\ref{sec:iso} by establishing a similar operational equivalence in the distillation of isotropic states, showing that any class of operations ranging from separable operations to PPT- and separability-preserving operations achieve the same one-shot rates of distillation. Analogously, in Sec.~\ref{sec:maxcorr} we show that separability-preserving operations provide no advantage over PPT operations in the distillation from maximally correlated states, and furthermore, by relating the entanglement monotones with measures of quantum coherence, the achievable rates and fidelities of distillation can be computed efficiently as semidefinite programs.

In Sec.~\ref{sec:assist} we show how our results immediately imply that in the setting of environment-assisted entanglement distillation~\cite{divincenzo_1999,smolin_2005,buscemi_2013}, all considered operations --- from one-way LOCC to PPT- and separability-preserving --- achieve exactly the same performance. We furthermore recover the one-shot characterisation of~\cite{buscemi_2013} in a simplified fashion by employing the formalism introduced herein.

We conclude in Sec.~\ref{sec:zero} with a discussion of zero-error distillation under the different sets of operations, obtaining in particular a single-letter formula for the asymptotic zero-error distillable entanglement under Rains-preserving operations which recovers a bound of Ref.~\cite{wang_2017-3} and endows it with an operational interpretation as a zero-error Rains bound.

\begin{figure}[t]
  \begin{center}
    \includegraphics[width=6cm]{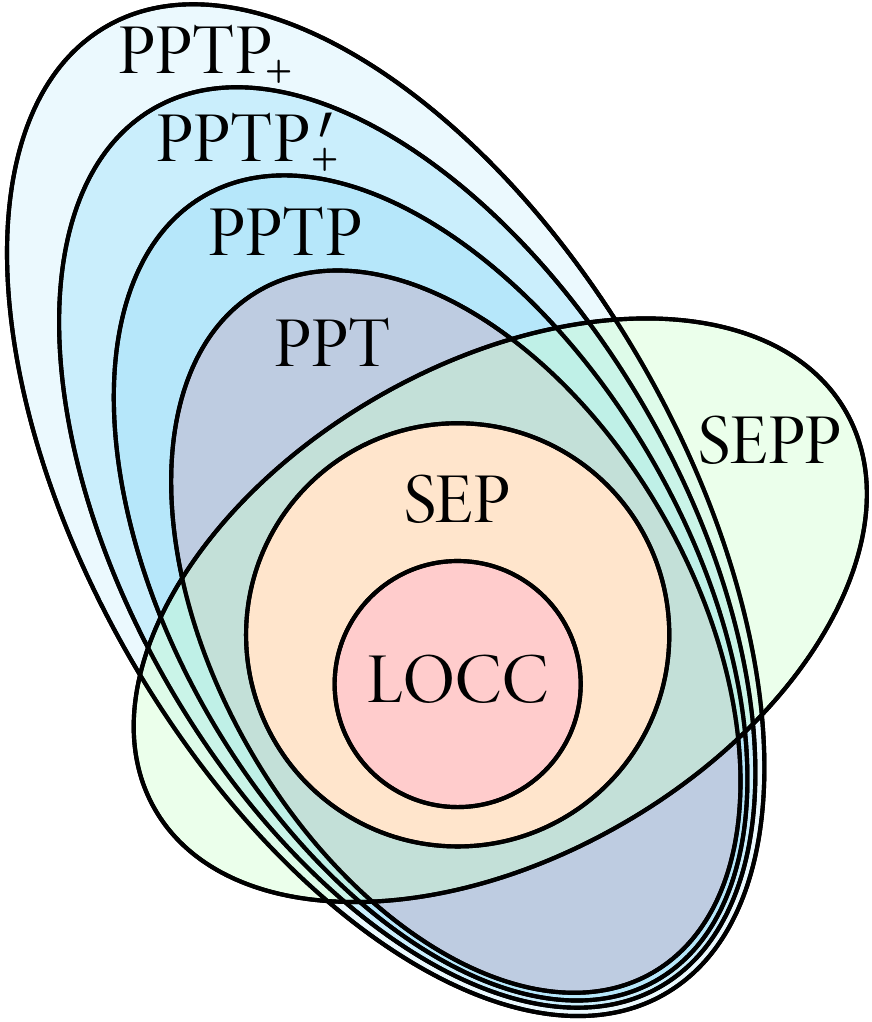}
    \begin{ruledtabular}
\begin{tabular}{ c c}
Class of operations &
Acronym \\
\colrule
Local operations and classical communication  &  LOCC\\
Separable operations & SEP\\
Separability-preserving operations  (see Sec.~\ref{sec:sepp}) & $\SEPP$ \\
 PPT operations (see Def.~\ref{def: PPT opt}) & $\PPTO$\\
 $\PPT$-preserving operations (see Def.~\ref{def:PPTPR}) & $\PPTPR$ \\
 Rains-preserving operations (see Sec.~\ref{sec:rains}) & $\PPTRPPR$ \\
$\PPTP$-preserving operations (see Def.~\ref{def: PPTPPR})& $\PPTPPR$\\
\end{tabular}
\end{ruledtabular}
    \caption{Schematic hierarchy of operations considered in this work. The pictured inclusions between $\PPTPPR$, $\PPTRPPR$, $\PPTPR$, $\PPTO$, $\SEPO$, and $\LOCC$ are all strict; there is no inclusion between $\SEPP$ and any of the sets $\PPTPPR$, $\PPTPR$, $\PPTO$ in general.}
    \label{fig:op-ent}
  \end{center}
  \vspace*{-2.6\baselineskip}
\end{figure}

\section{Preliminaries}\label{sec:prelim}  \vspace*{-.4\baselineskip}

We will work in the real vector space of Hermitian matrices $\HH$ with the Hilbert-Schmidt inner product $\<X,Y\> = \Tr(X^\dagger Y)$. We will denote by $\HH_+$ the cone of positive semidefinite matrices and by $\mgeq$ the inequality with respect to this cone, that is, $X \in \HH_+ \iff X \mgeq 0$. We will denote by $\HH_1$ the set of unit trace Hermitian matrices, and by $\DD = \HH_1 \cap \HH_+$ the set of density matrices. The notation $\ket{x}$ will be used to refer to general vectors in $\CC^d$, with Greek letters such as $\ket\psi$ reserved for normalised vectors corresponding to quantum states; in the latter case we will often refer to the projector $\proj{\psi}$ as $\psi$. We will use $\lnorm{\ket{x}}{p} = (\sum_i |x_i|^p)^{1/p}$ for the $p$-norms in $\CC^d$ and $\norm{X}{p} = \Tr\left( |X|^{p} \right)^{1/p}$ for the Schatten $p$-norms in $\HH$.

For any set $\Q$, we define the dual cone $\Q\* = \lset X \bar \<X, Q\> \geq 0 \; \forall Q \in \Q \rset$ and the polar set $\Q^\circ = \lset X \bar \<X, Q\> \leq 1 \;\forall\; Q \in \Q\rset$. We have in particular $\Q\*\* \coloneqq (\Q\*)\* = \cl \conv \lset \lambda Q \bar \lambda \geq 0,\; Q \in \Q \rset$ (the closure of the conic hull of $\Q$) and $\Q^{\circ\circ} \coloneqq (\Q^\circ)^\circ = \cl \conv ( \Q \cup \{0\} )$ where $\cl$ denotes closure and $\conv$ the convex hull of a set.

\begin{table}[t]
\caption{\label{tab:hierarchy}%
Comparison between the one-shot $\ve$-error entanglement distillation rates achievable under the different sets of operations considered in this work. We use $>$ when the given inequality can be strict for some state, and $\geq$ if --- to the best of our knowledge --- the strictness of the inequality remains an open question.
}
\begin{ruledtabular}
\begin{tabular}{ c c c c }
General states&
Pure states&
Isotropic states&
Max. corr. states\\
\colrule
$\SEPP$ & $\SEPP$ & $\SEPP$ & $\SEPP$ \\
\reflectbox{\rotatebox{90}{$<$}} & \rotatebox{90}{$=$} & \rotatebox{90}{$=$} & \rotatebox{90}{$=$}\\
$\PPTPPR$ & $\PPTPPR$ & $\PPTPPR$ & $\PPTPPR$\\
\reflectbox{\rotatebox{90}{$<$}}  & \rotatebox{90}{$=$} & \rotatebox{90}{$=$} & \rotatebox{90}{$=$}\\
$\PPTRPPR$ & $\PPTRPPR$ & $\PPTRPPR$ & $\PPTRPPR$\\
\reflectbox{\rotatebox{90}{$<$}}  & \rotatebox{90}{$=$} & \rotatebox{90}{$=$} & \rotatebox{90}{$=$}\\
$\PPTPR$ & $\PPTPR$ & $\PPTPR$ & $\PPTPR$\\
\rotatebox{90}{$=$} & \rotatebox{90}{$=$} & \rotatebox{90}{$=$} & \rotatebox{90}{$=$}\\
$\PPTO$ & $\PPTO$ & $\PPTO$ & $\PPTO$\\
\reflectbox{\rotatebox{90}{$<$}}  & \rotatebox{90}{$=$} & \rotatebox{90}{$=$} &~\reflectbox{\rotatebox{90}{$\leq$}} \\
$\SEPO$ & $\SEPO$ & $\SEPO$ & $\SEPO$ \\
\reflectbox{\rotatebox{90}{$\leq$}}  & \rotatebox{90}{$=$} &~\reflectbox{\rotatebox{90}{$\leq$}}  &~\reflectbox{\rotatebox{90}{$\leq$}} \\
(1-)$\LOCC$ & (1-)$\LOCC$ & (1-)$\LOCC$ & (1-)$\LOCC$\\
\end{tabular}
\end{ruledtabular}
\end{table}

All logarithms in this work are base $2$. We will use the shorthand
\begin{equation}\begin{aligned}
  \logfloor{x} \coloneqq \log \left\lfloor 2^x \right\rfloor,
\end{aligned}\end{equation}
and analogously for $\logceil{x}$.

\subsection{A family of entanglement monotones}

The analysis of this work will focus on understanding the achievable fidelity of distillation under different sets of operations, and establishing methods allowing us to relate it with convex optimisation problems which admit an efficient characterisation. To this end, we will introduce a family of entanglement monotones, which we will later explicitly endow with an operational interpretation and show to play a fundamental role in characterizing entanglement distillation.

Consider a bipartite system shared between parties $A$ and $B$, with $d_A$ and $d_B$ denoting the dimensions of the corresponding spaces. Let $d = \min\{d_A,d_B\}$. 
We will consider the following sets of Hermitian matrices:
\begin{equation}\begin{aligned}
  &\ppt = \lset X \bar \Tr(X) = 1,\; X^{T_B} \mgeq 0 \rset\\
  &\pptp = \lset X \bar \Tr(X) = 1,\; X \mgeq 0,\; X^{T_B} \mgeq 0 \rset\\
  &\SEP = \conv \lset \proj{\psi} \bar \ket\psi = \ket{\phi}_A \otimes \ket{\eta}_B \rset\\
\end{aligned}\end{equation}
where $X^{T_B}$ is the partial transpose of $X$. Letting $\S$ denote one of the above sets, a quantifier which found use in measuring the entanglement of quantum states in several contexts is the \textit{generalised robustness}, defined as~\cite{vidal_1999}
\begin{equation}\begin{aligned}
  R^\DD_{\S}(\rho) &= \min \lset \lambda \geq 0 \bar \rho \mleq (\lambda+1) X,\;X \in \S \rset\\
  &= \max \lset \< \rho, W \> \bar -\id \mleq W,\; W \in -\S\* \rset.
\end{aligned}\end{equation}
We then extend this definition to a class of measures
\begin{equation}\begin{aligned}\label{eq:family_Tm}
\TTT{m}_\S (\rho) &\coloneqq \max \lset \<\rho, W \> \bar -\id \mleq W \mleq m\id,\; W \in -\S\* \rset
\end{aligned}\end{equation}
for some parameter $m \in \RR_+$. We note this class of measures has been considered in~\cite{brandao_2005}, but we have found that some of the results concerning the quantification and characterisation of $\TTT{m}_\S$ stated there are in fact incorrect, so we present a self-contained investigation of their basic properties below and in the Appendix.

A useful characterisation of the quantifiers is obtained by considering their dual form, which can be obtained as follows.
\begin{proposition}\label{prop:proj_pos_neg} The measures $\TTT{m}_{\S}$ can be equivalently expressed as
\begin{equation}\begin{aligned}
  \TTT{m}_{\S} (\rho) = \min \lset m \Tr\left(\rho-X\right)_+ + \Tr\left(\rho-X\right)_- \bar X \in \S\*\* \rset,
\end{aligned}\end{equation}
where $(\rho-X)_+$ (respectively, $(\rho-X)_-$) denotes the positive (negative) part of the Hermitian operator $\rho-X$.
\end{proposition}
The proof follows well-known methods in matrix analysis and we include it in the Appendix for completeness.

In particular, for $m=1$ the measures take the remarkably simple form $\TTT{1}_{\S} (\rho) = \min_{X \in \S\*\*} \norm{\rho - X}{1}$. This quantity, considered first in the resource theory of coherence as the modified trace distance~\cite{yu_2016}, generalises the commonly employed trace distance measure $\min_{\sigma \in \S} \norm{\rho-\sigma}{1}$. The reason why $\TTT{1}_\S$ is a more suitable measure of entanglement than the trace distance itself is the fact that, contrary to $\TTT{1}_\S$, the trace distance does not satisfy strong monotonicity under LOCC~\cite{yu_2016,qiao_2018} (i.e. the requirement that a measure $M$ obeys $M(\rho) \geq \sum_i p_i M\left(\Lambda_i(\rho)\right)$ for any probabilistic protocol which applies an LOCC transformation $\Lambda_i$ to $\rho$ with probability $p_i$), which is often considered as one of the basic requirements that a measure of entanglement should satisfy~\cite{vidal_2000}. This demonstrates a case where it becomes necessary to consider the distance with respect to the unnormalised cone $\S\*\*$ rather than the set $\S$ in order to ensure strong monotonicity.

Another interesting case is $m=d-1$, for which we obtain the following.

\begin{proposition}\label{prop:rob_for_d-1}
For any $\S \in \{\PPT, \PPTP, \SEP\}$ it holds that
\begin{align}
  \TTT{d-1}_\S(\rho) &= R^\DD_\S(\rho).
\end{align}
\end{proposition}
\begin{proof}
Let $-W \in \SEP\*$, and notice that $\SEP \subseteq \PPTP \subseteq \PPT \Rightarrow \PPT\* \subseteq \PPTPd \subseteq \SEP\*$. Letting $\lambda_{\min}$ denote the smallest and $\lambda_{\max}$ the largest eigenvalue of a given matrix, we now use the property that $-W \in \SEP\* \Rightarrow \lambda_{\min}(-W) \geq (1 - d)\, \lambda_{\max}(-W)$~\cite[Cor. 5.5]{johnston_2010} together with the constraint $W \mgeq -\id$ to obtain
\begin{equation}\begin{aligned}
  &\lambda_{\max}(W) \leq d-1.
\end{aligned}\end{equation}
It follows that the feasible sets for $R^\DD_\S$ and $\TTT{d-1}_\S$ are equal and so the problems are equivalent.\end{proof}
From the above two Propositions, we have that the family of measures $\TTT{m}_\S$ can be understood as interpolating between the robustness of $\S$ for $m=d-1$, the modified trace distance for $m=1$, and the trivial value of $0$ for $m=0$. It is furthermore easy to see that $\TTT{m}_\S (\rho) = \TTT{d-1}_\S (\rho) $ for any $m > d-1$.

\subsection{Generalising $\TT{m}$ to arbitrary sets}

\begin{sloppypar}In the operational characterisation of entanglement distillation, it will be necessary to consider also generalisations of the above measures beyond sets of normalised (unit trace) Hermitian operators. To allow for this, we will now consider arbitrary compact sets of Hermitian operators $\Q$ and define the quantity
\begin{equation}\begin{aligned}
  \GG{m}_\Q(\rho) \coloneqq \sup \lsetr \< \rho, W \> \barr 0 \mleq W \mleq \id,\; W \in \frac{1}{m} \Q^\circ \rsetr.
\end{aligned}\end{equation}
It is straightforward to see that if $\Q$ is a set of unit trace operators, then $\GG{m}_\Q(\rho)$ is equal to $\frac{1}{m}\left(\TTT{m-1}_\Q(\rho)+1\right)$, although this relation does not hold in general.
\end{sloppypar}

To obtain a general dual formulation of $\GG{m}_\Q$, we will employ the formalism of gauge functions~\cite{rockafellar_1970,regula_2018}. The convex gauge function of a set $\Q$ is defined as~\cite{rockafellar_1970}
\begin{equation}\begin{aligned}
  \Gamma_\Q (\rho) &= \inf \lset \lambda \geq 0 \bar \rho \in \lambda \conv(\Q) \rset \\ &= \sup \lset \< \rho, W \> \bar W \in \Q^\circ \rset,
\end{aligned}\end{equation}
and can be thought of as an extension of the concept of a norm associated with a set --- indeed, all norms are gauge functions, but the latter are more general. One can further notice that
\begin{equation}\begin{aligned}\label{eq:gauge_dual}
  \Gamma_{\Q^\circ} (\rho) &= \inf \lset \lambda \geq 0 \bar \rho \in \lambda \Q^\circ \rset \\ &= \sup \lset \< \rho, W \> \bar W \in \Q^{\circ\circ} \rset
\end{aligned}\end{equation}
where the first equality follows because the set $\Q^\circ$ is always convex. We will take $\inf \varnothing = - \sup \varnothing = \infty$ and note that when the set $\Q$ is compact, the infima and suprema in the definitions of the gauge functions are attained as long as they are finite. We then have the following.

\begin{proposition}\label{prop:G_primal}
  For any compact set $\Q \subseteq \HH$, we have that
  \begin{equation}\begin{aligned}
    \GG{m}_\Q(\rho) = \inf_{Z \in \HH}\; \Tr(\rho - Z)_+ + \frac{1}{m} \Gamma_\Q(Z).
  \end{aligned}\end{equation}
\end{proposition}
\begin{proof}
The definition of $\GG{m}_\Q$ imposes that $W \in (-\HH_+)^\circ \cap \DD^\circ \cap (m \Q)^\circ = \left[(-\HH_+) \cup \DD \cup m \conv(\Q)\right]^\circ$, from which we get $\GG{m}_\Q = \Gamma_{(-\HH_+) \cup \DD \cup m \conv(\Q)}$. Since $\DD$ and $\Q$ are compact and $\HH_+$ is closed, we get~\cite[Thm. 16.4]{rockafellar_1970}
\begin{align}
  &\GG{m}_\Q(\rho) \nonumber\\
  &\quad = \inf \lset \Gamma_{-\HH_+}(X) + \Gamma_{\DD} (Y) + \Gamma_{m\Q} (Z) \bar \rho = X + Y + Z\rset\nonumber\\
  &\quad= \inf \lset \Gamma_{\DD} (Y) + \Gamma_{m\Q} (Z) \bar \rho = X + Y + Z,\; X \in -\HH_+ \rset\nonumber\\
  &\quad= \inf \lset \Tr Y + \frac{1}{m} \Gamma_{\Q} (Z) \bar \rho - Z \mleq Y,\; Y \in \HH_+ \rset\\
  &\quad= \inf \Tr (\rho - Z)_+ + \frac{1}{m} \Gamma_{\Q} (Z)\nonumber
\end{align}
where we have used that $\Gamma_{-\HH_+}(X) = 0$ if $X \in -\HH_+$ and $\infty$ otherwise, and $\Gamma_\DD(Y) = \Tr Y$ for any positive semidefinite $Y$ and $\infty$ if $Y$ is not positive  semidefinite.
\end{proof}

\begin{remark}The above formula effectively constraints the optimisation to be over $Z \in \Q\*\*$, since for any $Z \notin \Q\*\*$, we have $\Gamma_\Q(Z) = \infty$. In particular, if $\Q$ consists only of trace-one matrices, we can equivalently write
\begin{equation}\begin{aligned}
  \GG{m}_\Q(\rho) = \inf_{X \in \Q\*\*}\; \Tr(\rho - X)_+ + \frac{1}{m} \Tr(X)
\end{aligned}\end{equation}
which reduces to the form in Prop. \ref{prop:proj_pos_neg}.
\end{remark}


\section{One-shot entanglement distillation}\label{sec:dist}

Denoting by $\Psi_m$ the maximally entangled state $\ket{\Psi_m} = \sum_{i=1}^{m} \frac{1}{\sqrt{m}} \ket{ii}$, we consider the task of distilling the state $\Psi_m$ under a given class of completely positive trace-preserving (CPTP) maps $\O$. The fidelity of distillation under $\O$ is defined by
  \begin{align}
       F_{\O}(\rho,m) \coloneq \sup_{\Lambda \in \O} \< \Lambda(\rho), \Psi_m \>.
     \end{align}
Here, without loss of generality we constrain ourselves to operations $\Lambda \in \O$ whose output dimension matches the dimension of $\Psi_m$ in order to make the inner product well-defined; more general cases can be considered by suitably embedding $\Psi_m$ or $\Lambda(\rho)$ in a larger space. 
The one-shot $\ve$-error distillable entanglement is then defined as the maximum size of $\Psi_m$ which can be obtained with the given class of operations within an error tolerance of $\ve$, that is,
\begin{equation}
E_{d,\O}^{(1),\ve}(\rho) \coloneqq \log \max \lset m \in \mathbb N \bar F_\O(\rho,m) \geq 1- \ve \rset.
\label{one shot distillable ent}
\end{equation}
In the asymptotic i.i.d. limit, distillable entanglement can then be expressed as
\begin{align}
  E^\infty_{d,\O}(\rho) = \limsup_{\ve \to 0} \limsup_{n \to \infty} \frac{1}{n} E_{d,\O}^{(1),\ve}(\rho^{\otimes n}).
\end{align}

To begin the general description of one-shot distillation, we will make explicit the connection between the quantifiers discussed earlier and distillation rates. The precise link will be established through the \textit{hypothesis testing relative entropy}~\cite{buscemi_2010,wang_2012,tomamichel_2013}, defined as
 \begin{equation}\begin{aligned}
\label{hypothesis testing definition}
D_H^\ve(\rho||X) \coloneqq -\log\min \{  \< M, X\> \;&|\; 0\le M\le \id,\\
&1- \< M, \rho \> \le\ve \}.
\end{aligned}\end{equation}
where we have extended the standard definition (limited to positive semidefinite $X$) by taking $\log(x) = -\infty \; \forall x \leq 0$. This quantity characterises the fundamental task of quantum hypothesis testing \cite{hayashi_2016,hayashi_2017}, where one is interested in distinguishing between two quantum states --- $\rho$ and $\sigma$ --- by performing a test measurement $\{M, \id - M\}$ where $0 \mleq M \mleq \id$. The probability of incorrectly accepting state $\sigma$ as true (type-I error) is given by $\<\id - M, \rho\>$, and the probability of incorrectly accepting state $\rho$ as true (type-II error) is given by $\< M, \sigma \>$. The entropy $D_H^\ve(\rho||\sigma)$ then quantifies the minimum type-II error while constraining the type-I error to be no greater than $\ve$. We note that it is not clear if such an operational understanding of $D^\ve_H$ can be obtained when $X$ is not a positive semidefinite operator, but we will find it useful to consider the quantity $D^\ve_H$ regardless. Furthermore, we remark that for any operator $X$, $D_H^\ve(\rho||X)$ is efficiently computable as a semidefinite program.

Let us first note a general correspondence between the hypothesis testing relative entropy and gauge functions, showing that $D_H^\ve$ minimised over a set of operators gives a suitably ``smoothed'' gauge function.
\begin{proposition}\label{prop:hyp_gauge}
Let $\Q$ be a closed set of Hermitian operators. Then
\begin{equation}\begin{aligned}
\inf_{X \in \conv(\Q)} D^{\ve}_H(\rho\|X) = - \log \inf_{\substack{\< \rho, W \> \geq 1 - \ve\\0 \mleq W \mleq \id}} \Gamma_{\Q^\circ}(W).
\end{aligned}\end{equation}
\end{proposition}
\begin{proof}\begin{sloppypar}
We have
\begin{equation}\begin{aligned}
- \log \inf_{\substack{\< \rho, W \> \geq 1 - \ve\\0 \mleq W \mleq \id}} \Gamma_{\Q^\circ}(W)  &= - \log \inf_{\substack{\< \rho, W \> \geq 1 - \ve\\0 \mleq W \mleq \id}} \;\sup_{X \in \Qcc} \< X, W \>\\
   &= - \log \sup_{X \in \Qcc} \inf_{\substack{\< \rho, W \> \geq 1 - \ve\\0 \mleq W \mleq \id}} \< X, W \>\\
   &= - \log \sup_{X \in \conv(\Q)} \inf_{\substack{\< \rho, W \> \geq 1 - \ve\\0 \mleq W \mleq \id}} \< X, W \> \\
   &= \inf_{X\in\conv(\Q)} - \log \inf_{\substack{\< \rho, W \> \geq 1 - \ve\\0 \mleq W \mleq \id}} \< X, W \> \\
   &= \inf_{X \in \conv(\Q)} D^{\ve}_H(\rho\|X)
\end{aligned}\end{equation}
where the second equality follows by Sion's minimax theorem, since the sets $\big\{ W \;\big|\; \< \rho, W \> \geq 1 - \ve,\; 0 \mleq W \mleq \id \big\}$ and $\Qcc$ are both convex and the former is compact. We have replaced the optimisation over $\Q^\cc$ with an optimisation over $\conv(\Q)$ without loss of generality, since the problem has the same optimal value in both cases --- either there exists an $X \in \conv(\Q)$ such that $D^{\ve}_H(\rho\|X) < \infty$, or we have $D^{\ve}_H(\rho\|X) = D^{\ve}_H(\rho\|0) = \infty$ for all $X \in \Q^\cc$. \qedhere
\end{sloppypar}\end{proof}

The above can be directly applied in the context of entanglement distillation. Specifically, if one can show that the fidelity of distillation under a given class of operations is given by $\GG{m}_\Q$ for some set $\Q$, then the optimal rate of distillation can be computed exactly as the hypothesis testing entropy minimised over $\conv(\Q)$. Although we leave open the question of when exactly a given class of operations leads to a fidelity of distillation of the form given by $\GG{m}_\Q$, we will see below that this is a very common phenomenon among different classes of operations relevant to the resource theory of entanglement.

Formally, we have the following.
\begin{theorem}\label{thm:hyp_entropy}
Let $\O$ be a class of CPTP operations, and $\Q$ a compact set of Hermitian operators. If a given state $\rho$ satisfies 
\begin{equation}\begin{aligned}
  F_\O(\rho, m) = \GG{m}_\Q(\rho) \quad \forall \, m \in \NN,
\end{aligned}\end{equation}
then
\begin{equation}\begin{aligned}
  E^{(1),\ve}_{d,\O} (\rho) &= \logfloor{ \min_{X \in \conv(\Q)} D^{\ve}_H(\rho\|X) }.
\end{aligned}\end{equation}
\end{theorem}
 \begin{remark}The theorem includes in particular the case when $F_{\O}(\rho, m) = \frac{1}{m}\left(\TTT{m-1}_\S(\rho)+1\right)$ for $\S \in \{\SEP,\PPT,\PPTP\}$. However, it is more general than that --- for example, $\Q$ can be the set $\PPTR =\lset X \bar \norm{X^{T_B}}{1} \leq 1 \rset$, in which case we recover a result of~\cite{fang_2019}.\end{remark}
\begin{proof}
By assumption, we have
\begin{equation}\begin{aligned}
 & E^{(1),\ve}_{d,\O} (\rho) \\
 &= \log \max \Big\{ m \in \NN \;\Big|\; \< \rho, W \> \geq 1 - \ve,\\ 
 & \hphantom{= \log \max \Big\{ m \in \NN \;\Big| } \; 0 \mleq W \mleq \id,\; W \in \frac{1}{m} \Qc \Big\} \\ 
 &= \Big\lfloor - \log \min \big\{ k \in \RR \;\big|\; \< \rho, W \> \geq 1 - \ve,\\
 & \hphantom{= \big\lfloor - \log \min \big\{ k \in \RR \;\big|} 0 \mleq W \mleq \id,\; W \in k \Qc \big\}\Big\rfloor_{\log}\\
 &= \logfloor{ - \log \min_{\substack{\< \rho, W \> \geq 1 - \ve\\0 \mleq W \mleq \id}} \Gamma_{\Qc} (W) }\\
 &= \logfloor{  \min_{X \in \conv(\Q)} D^{\ve}_H(\rho\|X) }
\end{aligned}\end{equation}
where the last equality follows from Prop.~\ref{prop:hyp_gauge}.\end{proof}

The application of the above result will allow us to employ the powerful framework of convex optimisation in the description of entanglement distillation.


\subsection{PPT and PPT-preserving operations}\label{sec:dist_ppt}

One of the first relaxations of LOCC in the literature was the class of separable operations (SEP)~\cite{vedral_1997,rains_1997}, corresponding to all quantum channels $\Lambda: AB \to A'B'$ whose Choi matrix is separable across the bipartition $AA'|BB'$. This set of maps has been shown to be strictly larger than LOCC~\cite{bennett_1999}, thus providing an upper bound on the capabilities of LOCC in distillation. However, the fact that the definition of $\SEPO$ relies on the separability of the Choi matrix means that the set is not amenable to an efficient analytical characterisation, which then motivated the definitions of larger sets of operations. We begin with the investigation of several classes of such operations based on the set $\PPT$.

The class of PPT operations, due to Rains~\cite{rains_1999-1,rains_2001}, is defined to consist of all CPTP maps $\Lambda: AB \to A'B'$ whose Choi matrix $J_\Lambda$ satisfies $J^{T_{BB'}}_\Lambda \mgeq 0$. In some works, a closely related class of ``PPT-preserving operations'' has been considered~\cite{eggeling_2001,audenaert_2003}, motivated by the fact that $X \in \PPT \Rightarrow \Lambda(X) \in \PPT$ for any PPT operation $\Lambda$. Although the two classes have sometimes been claimed to be equal, it is not difficult to see that only imposing the PPT-preserving constraint leads to a strictly larger class of quantum channels --- consider, for instance, the channel which swaps subsystems $A$ and $B$ --- so it is in fact incorrect to use the names ``PPT'' and ``PPT-preserving'' interchangeably when referring to operations.
Interestingly, however, the two sets of channels lead to exactly the same rates of one-shot entanglement distillation (as well as dilution), as we will shortly see explicitly.

More recently, the name ``PPT-preserving operations'' was also used to denote operations which map any PPT state to a PPT state, in the sense that $\sigma \in \PPTP \Rightarrow \Lambda(\sigma) \in \PPTP$~\cite{chitambar_2017}. It is well-known that this leads to a strictly larger class of operations than Rains' PPT operations~\cite{horodecki_2001}, although an accurate way of referring to the class of PPT operations could be \textit{completely} PPT-preserving~\cite{matthews_2008}, since the condition $J^{T_{BB'}}_\Lambda \mgeq 0$ ensures the preservation of positivity when the map acts on a part of a larger system, akin to completely positive maps. 


For clarity, let us begin with the precise definitions.%
\begin{definition}\label{def: PPT opt}
A CPTP map $\Lambda: AB \to A'B'$ is called \textbf{PPT} if any one of the following equivalent conditions is satisfied~\cite{rains_1999-1}:
\begin{enumerate}[label=(\roman*)]
\item The Choi matrix $J_\Lambda$ is PPT with respect to the partition $AA'|BB'$, i.e. $J^{T_{BB'}}_\Lambda \mgeq 0$.
\item The map $T_{B'} \circ \Lambda \circ T_{B}$ is completely positive.
\item For any spaces $C,D$ such that $d_C =d_A$ and $d_D = d_B$ it holds that
\begin{equation*}\begin{aligned}
  \sigma \in \PPTP(AC|BD) \Rightarrow \Lambda \otimes \idc (\sigma) \in \PPTP(A'C|B'D)
\end{aligned}\end{equation*}
where $\idc$ is the identity channel.
\end{enumerate}
We will use $\PPTO$ to denote the set of all such maps.
\end{definition}%
\begin{definition}\label{def: PPTPPR}
A CPTP map $\Lambda: AB \to A'B'$ is called \textbf{$\PPTP$-preserving} if
\begin{equation}\begin{aligned}
  \sigma \in \PPT_+ \Rightarrow \Lambda(\sigma) \in \PPT_+.
\end{aligned}\end{equation}
We will use $\PPTPPR$ to denote the set of all such maps.
\end{definition}%
\begin{definition}\label{def:PPTPR}
A CPTP map $\Lambda: AB \to A'B'$ is called \textbf{$\PPT$-preserving} if either of the following equivalent conditions is satisfied:
\begin{enumerate}[label=(\roman*)]
\item $\begin{displaystyle}X \in \PPT \Rightarrow \Lambda(X) \in \PPT\end{displaystyle}$.
\item The map $T_{B'} \circ \Lambda \circ T_{B}$ is positive, i.e. $X \in \HH_+ \Rightarrow T_{B'} \circ \Lambda \circ T_{B} (X) \in \HH_+$.
\end{enumerate}
We will use $\PPTPR$ to denote the set of all such maps.
\end{definition}%

We now characterise the operational capabilities of the different sets of operations. Note that the fidelity of distillation under the class $\PPTO$ has previously been obtained by Rains~\cite{rains_2001}, and an explicit expression for the rate of distillation in terms of $D^{\ve}_H$ appeared more recently in~\cite{fang_2019}.~\cite{eggeling_2001} considered the class $\PPTPR$ in this context, but the capabilities of $\PPTPPR$ have not been explicitly investigated before.
\begin{theorem}\label{thm:ppt_dist}The fidelity of distillation under the classes of operations $\PPTO$, $\PPTPR$, and $\PPTPPR$ is given by
\begin{equation}\begin{aligned}
F_{\PPTPPR} (\rho, m) &= \GG{m}_{\PPTP} (\rho)\\
  F_\PPTO (\rho,m) = F_{\PPTPR} (\rho,m) &= \GG{m}_{\PPTR} (\rho)
\end{aligned}\end{equation}
where $\PPTR = \conv( \PPT \cup - \PPT) = \lset X \bar \norm{X^{T_B}}{1} \leq 1 \rset$, and hence the one-shot distillable entanglement can be expressed as
\begin{align}
E^{(1),\ve}_{d,\PPTPPR} (\rho) &= \logfloor{ \min_{\sigma \in \PPTP} D^{\ve}_H(\rho\|\sigma) } \label{distillation PPTPplus}\\
  E^{(1),\ve}_{d,\PPTO} (\rho) = E^{(1),\ve}_{d,\PPTPR} (\rho) &= \logfloor{ \min_{X \in \PPTR} D^{\ve}_H(\rho\|X) }. \label{distillation PPT PPTP}
\end{align}
\end{theorem}

\begin{proof}Since $\Psi_m$ is invariant under any unitary of the form $U \otimes U^*$, it is in particular invariant under the twirling $\T(\cdot) \coloneqq \int (U \otimes U^*) \cdot (U \otimes U^*)^\dagger \,\mathrm{d}U$, where the integration is performed with respect to the Haar measure of the unitary group. We can then without loss of generality consider only trace-preserving operations of the form $\Lambda = \T \circ \Lambda$, giving
\begin{equation}\begin{aligned}\label{eq:proof:lambda_twirl}
  \Lambda(Z) =  \< Z, X\> \Psi_m + \frac{\< Z, \id - X \>}{m^2-1} \left(\id - \Psi_m\right)
\end{aligned}\end{equation}
as this is the most general form of an operator invariant under twirling~\cite{horodecki_1999}. Since $\Psi^{T_B}_m = \frac1m ( P^+_m - P^-_m )$, where $P^+_m$ (respectively, $P^-_m$) denote the projector onto the symmetric (antisymmetric) subspace, we have
\begin{equation}\begin{aligned}
    \Lambda(Z)^{T_B} &= \frac{P^+_m}{m+1} \left( \frac{\Tr Z}{m} + \< Z, X \>\right) \\&+ \frac{P^-_m}{m-1} \left( \frac{\Tr Z}{m} - \< Z, X \>\right).
\end{aligned}\end{equation}
Using the mutual orthogonality of $P^\pm_m$, we obtain the general conditions
\begin{equation}\begin{aligned}
  \Lambda(Z) \mgeq 0 &\iff \<Z, X \> \geq 0\\
  \Lambda(Z)^{T_B} \mgeq 0 &\iff -\frac{\Tr Z}{m} \leq \< Z, X \> \leq \frac{\Tr Z}{m}.
\end{aligned}\end{equation}
Noting in addition that the complete positivity of $\Lambda$ imposes $0 \mleq X \mleq \id$, we can constrain the map $\Lambda$ such that $\Lambda(\sigma) \in \PPTP$ for any $\sigma \in \PPTP$ to get
\begin{equation}\begin{aligned}
  &F_{\PPTPPR}(\rho,m) \\
  &= \max \lsetr \< \rho, X \> \barr 0 \mleq X \mleq \id,\; \< \sigma, X \> \leq \frac{1}{m} \; \forall \sigma \in \PPT_+ \rsetr,
\end{aligned}\end{equation}
and similarly, by imposing that $S \in \PPT \Rightarrow \Lambda(S) \in \PPT$ we have
\begin{equation}\begin{aligned}
  &F_{\PPTPR}(\rho,m) \\
  &= \max \lsetr \< \rho, X \> \barr 0 \mleq X \mleq \id,\; \left|\< S, X \>\right| \leq \frac{1}{m} \; \forall S \in \PPT \rsetr\\
  &= \max \lsetr \< \rho, X \> \barr 0 \mleq X \mleq \id,\; \< S, X \> \leq \frac{1}{m} \; \forall S \in \PPTR \rsetr.
\end{aligned}\end{equation}
Noting that the Choi matrix of the map \eqref{eq:proof:lambda_twirl} is given by $J_\Lambda = X_{AB} \otimes \Psi_{m_{A'B'}} + \frac{1}{m^2-1}\left(\id - X\right)_{AB} \otimes (\id - \Psi_m)_{A'B'}$, an explicit computation yields
\begin{equation}\begin{aligned}
  J^{T_{BB'}}_\Lambda \mgeq 0 \iff - \frac{1}{m} \id \mleq X^{T_B} \mleq\frac{1}{m} \id,
\end{aligned}\end{equation}
which is precisely the condition $\< S, X \> \leq \frac{1}{m} \; \forall S \in \PPTR$, yielding the equality between $F_{d,\PPTPR}(\rho,m)$ and $F_{d,\PPTO}(\rho,m)$.

Since the distillation fidelities $F_{\PPTPPR}$ and $F_{\PPTPR}$ are precisely of the form $G_\PPTP$ and $G_\PPTR$, respectively, the result follows by Thm.~\ref{thm:hyp_entropy}.
\end{proof}

The Theorem establishes an operational equivalence between the sets of operations $\PPTO$ and $\PPTPR$, although we stress again that in fact $\PPTO \subsetneq \PPTPR$: in particular, the swap operation, defined as $\Lambda(\ketbra{ij}{kl}) = \ketbra{ji}{lk}$ in a basis and extended by linearity, trivially preserves the positivity of the partial transpose of any operator, while the partial transpose of the Choi matrix $J_\Lambda$ can be verified to be non-positive. This can be understood by noting that the swap operation does not preserve PPT states when acting only on a part of a larger system --- indeed, if Alice and Bob each possess a singlet and exchange only half of it, they will have generated (maximal) entanglement. Notice also that we have explicitly shown a difference between the distillation rates of $\PPTPR$ and $\PPTPPR$, thus immediately implying that $\PPTPR \subsetneq \PPTPPR$.

In addition, we recall an argument in~\cite{chitambar_2017} which investigated a gap between $\PPTO$ and $\PPTPPR$ operations by showing that the negativity (a known monotone under $\PPTO$~\cite{vidal_2002}) can increase under $\PPTPPR$. This argument no longer applies to $\PPTPR$ --- the negativity can be expressed as a robustness-type quantifier with respect to the set $\PPT$~\cite{vidal_2002} and it follows straightforwardly that this is a strong monotone under $\PPTPR$~\cite{regula_2018}. The gap between $\PPTPR$ and $\PPTO$ is therefore much more subtle.

Although it is not easy to characterise the asymptotic rates of distillation under $\PPTO$ and $\PPTPR$ maps, we have the following characterisation of distillable entanglement under $\PPTPPR$, thus establishing a limit on the asymptotic performance of $\PPTO$ and $\PPTPR$ (see also~\cite{brandao_2015}).
\begin{corollary}\label{cor:ppt_asymptotic}
The asymptotic distillable entanglement under $\PPTPPR$ is given by the regularised relative entropy of entanglement with respect to the set $\PPTP$,
\begin{equation}\begin{aligned}
  E^{\infty}_{d,\PPTPPR} (\rho) &= E^{\infty}_{R,\PPTP} (\rho) \coloneqq \lim_{n\to\infty} \min_{\sigma \in \PPTP} \frac{1}{n} D(\rho^{\otimes n}\|\sigma)
\end{aligned}\end{equation}
with $D$ denoting the quantum relative entropy.
\end{corollary}
\begin{proof}
Follows directly from the generalised quantum Stein's lemma~\cite{brandao_2010-1}, which shows precisely that the regularisation of the hypothesis testing relative entropy $\min_{\sigma \in \PPTP} D^{\ve}_H(\rho\|\sigma)$ in the asymptotic limit with $\ve$ going to $0$ is given by the regularised relative entropy.
\end{proof}
We remark that, although $E^{\infty}_{d,\PPTPPR} (\rho)$ is not known in general, it has been computed exactly for classes of all orthogonally invariant states (including isotropic and Werner states)~\cite{audenaert_2002}, and it has been shown that there exist states such that $E^{\infty}_{d,\PPTO} (\rho) < E^{\infty}_{d,\PPTPPR} (\rho)$~\cite{wang_2017-1}. 

\subsubsection{Rains set and distillation}\label{sec:rains}

The above Corollary in particular gives an operational interpretation to the regularised relative entropy $E^{\infty}_{R,\PPTP}$, introduced first as a bound for distillable entanglement in~\cite{rains_1999-1}. One can then wonder whether similar operational interpretation can be given to other asymptotic quantities in entanglement distillation theory. We will show that it is indeed the case for one of the most fundamental of such bounds, the regularised Rains bound~\cite{rains_2001,audenaert_2002,wang_2017-3}, constituting the tightest known bound for the asymptotically distillable entanglement. It is defined as
\begin{equation}\begin{aligned}
  E^{\infty}_{\textrm{Rains}} (\rho) \coloneqq \lim_{n\to\infty} \min_{X \in \PPTRP} \frac{1}{n} D(\rho^{\otimes n}\|X)
\end{aligned}\end{equation}
where $\PPTRP = \lset X \mgeq 0 \bar \norm{X^{T_B}}{1} \leq 1 \rset$ is the so-called Rains set. To relate this quantity with the distillation of entanglement, we will define the class of \textbf{Rains-preserving operations} $\PPTRPPR$ as all maps such that $X \in \PPTRP \Rightarrow \Lambda (X) \in \PPTRP$. We then have the following.

\begin{theorem}\label{thm:rains_preserving}
The fidelity of distillation and one-shot distillable entanglement under Rains-preserving operations $\PPTRPPR$ are given by
\begin{equation}\begin{aligned}
  F_{\PPTRPPR} (\rho, m) &= \GG{m}_{\PPTRP} (\rho)\\
  E^{(1),\ve}_{d,\PPTRPPR} (\rho) &= \logfloor{ \min_{X \in \PPTRP} D^{\ve}_H(\rho\|X) }.
\end{aligned}\end{equation}
\end{theorem}
\begin{proof}
The proof proceeds analogously to Thm.~\ref{thm:ppt_dist}. The crucial step is to notice that for the isotropic operator $\Lambda(Z) =  \< Z, X\> \Psi_m + \frac{\< Z, \id - X \>}{m^2-1} \left(\id - \Psi_m\right)$ we have
\begin{equation}\begin{aligned}
  \norm{\Lambda(Z)^{T_B}}{1} &= \frac{m}{2} \left| \frac{\Tr Z}{m} + \< Z, X \> \right| \\&
  + \frac{m}{2} \left| \frac{\Tr Z}{m} - \< Z, X \> \right|\
\end{aligned}\end{equation}
where we have used that $\frac{1}{m+1} \norm{P^+_m}{1} = \frac{1}{m-1} \norm{P^-_m}{1} =\frac{m}{2}$ and that $P^\pm_m$ are mutually orthogonal projections. For any $Z \in \PPTRP$, it is then easy to verify that we have $\norm{\Lambda(Z)^{T_B}}{1} \leq 1 \iff \< Z, X \> \leq \frac{1}{m}$. This gives
\begin{equation}\begin{aligned}
  &F_{\PPTRPPR}(\rho,m) \\
  &= \max \lsetr \< \rho, X \> \barr 0 \mleq X \mleq \id,\; \< Z, X \> \leq \frac{1}{m} \; \forall Z \in \PPTRP \rsetr\\
  &= \GG{m}_{\PPTRP} (\rho).
\end{aligned}\end{equation}
The statement about $E^{(1),\ve}_{d,\PPTRPPR}$ then follows directly from Thm.~\ref{thm:hyp_entropy}.
\end{proof}
Once again, an application of the generalised quantum Stein's lemma~\cite{brandao_2010-1} then gives
\begin{equation}\begin{aligned}
    E^{\infty}_{d,\PPTRPPR} (\rho) = E^{\infty}_{\textrm{\rm Rains}} (\rho),
\end{aligned}\end{equation}
which establishes an explicit operational interpretation of the regularised Rains bound as the asymptotic rate of entanglement distillation under the class of Rains-preserving operations. Noting that $\PPTRPPR \supseteq \PPTPR$ by definition, we recover the result that  $E^{\infty}_{\textrm{\rm Rains}}$ upper bounds the asymptotic distillable entanglement under $\PPTO$ operations~\cite{rains_2001}. It is interesting to conjecture that we have equality between $E^\infty_{d,\PPTO}$ and $E^{\infty}_{d,\PPTRPPR}$ (cf.~\cite{fang_2019}), but we were not able to establish this.

To obtain a tighter bound on distillable entanglement, one could then ask about distillation under operations which \textit{completely} preserve the Rains set, in the sense that 
\begin{equation}\begin{aligned}
  X \in \PPTRP(AC|BD) \Rightarrow \Lambda \otimes \textrm{id } (X) \in \PPTRP(A'C|B'D)
\end{aligned}\end{equation}
for some spaces $C,D$ such that $d_C = d_A$ and $d_D=d_B$. We will call any such channel \textbf{completely Rains-preserving}.
In other words, a map $\Lambda: AB \to A'B'$ is completely Rains-preserving iff it is CPTP and
\begin{equation}\begin{aligned}\label{eq:crp}
  \norm{ \Lambda \otimes \textrm{id} (X)^{T_{B'D'}} }{1} \leq 1 \quad \forall \, X: X \mgeq 0,\,  \norm{X^{T_{BD}}}{1} \leq 1.
\end{aligned}\end{equation}
We will now show that these maps are precisely the set of PPT channels.
\begin{theorem}
A quantum channel is PPT iff it is completely Rains-preserving.
\end{theorem}
\begin{proof}
One direction is straightforward: if $\Lambda$ is completely positive and completely Rains-preserving, then for any $\sigma \in \PPTP(AC|BD)$ we necessarily have $\Lambda \otimes \textrm{id } (\sigma) \in \PPTP(A'C'|B'D')$ due to the fact that $\PPTP = \PPTRP \cap \HH_+$. This means that $\Lambda$ is completely $\PPTP$-preserving, i.e. PPT.

To see the opposite inclusion, define a ``PPT-diamond norm'' of any map $\Gamma$ as
\begin{equation*}\begin{aligned}
  \norm{\Gamma}{\!\!\vardiamond} \coloneqq&\max \lset \norm{ \Gamma \otimes \textrm{id } (X^{T_{BD}}) }{1} \bar \norm{X^{T_{BD}}}{1} \leq 1,\; X \mgeq 0 \rset
  \\=& \max \lset \norm{ \Gamma \otimes \textrm{id } (X) }{1} \bar \norm{X}{1} \leq 1,\; X^{T_{BD}} \mgeq 0 \rset.
\end{aligned}\end{equation*}
Rewriting Eq. \eqref{eq:crp} one can see that $\Lambda$ is completely Rains-preserving iff
\begin{equation*}\begin{aligned}
  \norm{T_{B'} \circ \Lambda \circ T_B}{\!\!\vardiamond} \leq 1.
\end{aligned}\end{equation*}
Notice then that, for any Hermiticity-preserving map $\Gamma$ it holds that
\begin{equation*}\begin{aligned}
  \norm{\Gamma}{\!\!\vardiamond} &\leq \max \lset \norm{ \Gamma \otimes \textrm{id } (X) }{1} \bar \norm{X}{1} \leq 1 \rset = \norm{\Gamma}{\!\!\Diamond}
\end{aligned}\end{equation*}
where $\norm{\cdot}{\!\!\Diamond}$ is the diamond norm (completely bounded trace norm)~\cite{paulsen_2002,watrous_2004}.
Since for any PPT channel $\Lambda$ the map $T_{B'} \circ \Lambda \circ T_B$ is CPTP, we have that any PPT channel satisfies $\norm{T_{B'} \circ \Lambda \circ T_B}{\!\!\Diamond} = 1$~\cite{watrous_2004} and therefore is completely Rains-preserving.
\end{proof}
The above result establishes an operational connection between the sets $\PPTPPR$ and $\PPTRPPR$, showing that their ``completely preserving'' variants reduce to the same set of operations (PPT).


\subsection{Pure-state distillation and separability-preserving operations}\label{sec:sepp}

The class of \textbf{separability-preserving operations} $\SEPP$ is defined as all CPTP maps $\Lambda$ such that $\sigma \in \SEP \Rightarrow \Lambda(\sigma) \in \SEP$, that is, as the maximal class of free (non-entangling) operations in the resource theory of entanglement. Notice that this class does not \textit{completely} preserve separability, in the sense that it could generate entanglement if applied to a part of a larger system; if such complete preservation is imposed, we instead recover the class of separable operations. The inclusions between the different classes of operations are shown in Fig.~\ref{fig:op-ent}.

The fidelity of distillation under $\SEPP$ was first derived in~\cite{brandao_2010}, and can be used to characterise the distillable entanglement as follows.
\begin{lemma}[\cite{brandao_2011}]\label{lemma:sepp_fid}It holds that $\begin{displaystyle}\FF{\SEPP} (\rho,m) = \GG{m}_{\SEP} (\rho) = \frac{1}{m} \left(\TTT{m-1}_\SEP(\rho) + 1\right)\end{displaystyle}$, and hence
\begin{equation}\begin{aligned}
  E^{(1),\ve}_{d,\SEPP} (\rho) = \logfloor{ \min_{\sigma \in \SEP} D^{\ve}_H(\rho\|\sigma) }.
\end{aligned}\end{equation}\end{lemma}
\begin{proof}Follows in exactly the same way as the proof of Thm.~\ref{thm:ppt_dist}, since isotropic states of the form
\begin{equation}\begin{aligned}
  \< Z, X\> \Psi_m + \frac{\< Z, \id - X \>}{m^2-1} \left(\id - \Psi_m\right)
\end{aligned}\end{equation}
are separable if and only if they are PPT~\cite{horodecki_1999}.\end{proof}

By the inclusion $\SEP \subseteq \PPTP \subset \PPTR$, we immediately have that
\begin{equation}\begin{aligned}\label{eq:fid_and_G}
F_{\SEPP} (\rho,m) 
&\geq F_{\PPTPPR} (\rho,m)\\
&\geq F_{\PPTO} (\rho,m)
 \end{aligned}\end{equation}
 thus establishing a hierarchy of rates of distillation between the operations $\SEPP$, $\PPTPPR$, and $\PPTO$. Notice that this does not follow from their definition, as there is no inclusion between the sets of maps $\SEPP$ and $\PPTO$, nor between $\SEPP$ and $\PPTPPR$.

Crucially, for any pure state, the fidelity of distillation can be computed exactly. To establish this result, we will employ the so-called \textit{$m$-distillation norm}, introduced in~\cite{regula_2017} as
\begin{align}\label{eq:distillation_norm_def}
  \mnorm{\ket{x}} \coloneqq& \min_{\ket{x}=\ket{y}+\ket{z}} \lnorm{\ket{y}}{1} + \sqrt{m} \lnorm{\ket{z}}{2}\\
  =& \max \lset \cbraket{x|w} \bar \lnorm{\ket{w}}{\infty} \leq 1,\; \lnorm{\ket{w}}{2} \leq \sqrt{m} \rset\nonumber
\end{align}
for any vector $\ket{x} \in \CC^{d}$. One can immediately notice from the inequality $\lnorm{\cdot}{2} \leq \lnorm{\cdot}{1} \leq \sqrt{d} \lnorm{\cdot}{2}$ that we have $\norm{\ket x}{[1]} = \lnorm{\ket{x}}{2}$ and $\norm{\ket{x}}{[d]} = \lnorm{\ket{x}}{1}$. Notably, for any normalised vector $\ket{x} \in \CC^{d}$ and any integer $m \in \{1, \ldots, d\}$, the norm admits an exact expression as~\cite{regula_2017}
\begin{equation*}\begin{aligned}
  \mnorm{\ket{x}} = \lnorm{\ket{x^\downarrow_{1:m-k^\star}}}{1} + \sqrt{k^\star} \lnorm{\ket{x^\downarrow_{m-k^\star+1:d}}}{2},
\end{aligned}\end{equation*}
where $\ket{x^\downarrow_{1:k}}$ denotes the vector consisting of the $k$ largest (by magnitude) coefficients of $\ket{x}$, analogously $\ket{x^\downarrow_{k+1:d}}$ denotes the vector of the $d-k$ smallest coefficients of $\ket{x}$ with $\ket{x^\downarrow_{1:0}}$ being the zero vector, and we define
\begin{equation}\begin{aligned}
  k^\star \coloneqq \argmin_{1 \leq k \leq m} \frac{1}{k}\lnorm{\ket{x^\downarrow_{m-k+1:d}}}{2}^2.
\end{aligned}\end{equation}
We stress that the computation of $\mnorm{\ket{x}}$ is thus reduced to evaluating $m-1$ inequalities.

We will now use $\ket{\xi_\psi} \in \RR^d$ to denote the vector of Schmidt coefficients of a pure state $\ket\psi$, in the sense that $\ket{\xi_\psi} = (\alpha_1, \ldots, \alpha_d)^T$ where $\ket{\psi} = \sum_i \alpha_i \ket{i}_A \ket{i}_B$ for some orthonormal bases $\{\ket{i_A}\}, \{\ket{i}_B\}$. Employing the $m$-distillation norm, we then have the following.
\begin{theorem}\label{thm:Tm_pure_states}
For any $m \geq 1$, it holds that
\begin{equation}\begin{aligned}
  \TTT{m-1}_\SEP(\psi) = \mnorm{\ket{\xi_\psi}}^2 - 1,
\end{aligned}\end{equation}
and in particular $F_{\SEPP} (\psi,m) = \frac{1}{m} \mnorm{\ket{\xi_\psi}}^2$.
\end{theorem}
\begin{proof}To begin, notice that with a simple rearrangement of terms $\TTT{m-1}_\SEP$ can be written as
\begin{equation}\begin{aligned}
  &\max \lset \<\rho, W \> \bar 0 \mleq W \mleq m\id,\; W \in \SEP^\circ \rset\\
   = &\max \lsetr \< \rho, W \> \barr  W \mgeq 0,\;  W \in \left(\SEP \cup \frac{1}{m}\DD\right)^\circ \rsetr
\end{aligned}\end{equation}
where we used that $(\mathcal{C} \cup \mathcal{D})^\circ = \mathcal{C}^\circ \cap \mathcal{D}^\circ$ for convex and closed sets. The set $\conv(\SEP \cup \frac{1}{m}\DD) \eqqcolon \Q_m$ can be noticed to be the convex hull of rank-one terms as $\Q_m = \conv \lset \proj{x} \bar \ket{x} \in \V \cup \N_m \rset$ where $\V = \{ \ket{\phi}_A \otimes \ket{\eta}_B \}$ is the set of all normalised product state vectors and $ \N_m \coloneqq \lset \ket{x} \bar \lnorm{\ket{x}}{2} = 1 / \sqrt{m} \rset.$

By Thm. 10 in~\cite{regula_2018}, for any pure state $\ket\psi$ we then have~\footnote{Strictly speaking,~\cite[Thm. 10]{regula_2018} is obtained for sets of \textit{normalised} operators; it is easy to notice, however, that the proof does not rely on normalisation and the Theorem applies in full generality also for unnormalised sets of operators, such as $\Q_m$ in our proof.}
\begin{equation}\begin{aligned}
  \TTT{m-1}_\SEP (\psi) + 1 &= \max \lset \braket{\psi | W | \psi} \bar W \mgeq 0, W \in \Q_m^\circ \rset \\ &= \max \lset \cbraket{\psi | w}^2 \bar \ket{w} \in (\V \cup \N_m)^\circ \rset\\
  &= \Gamma_{\V \cup \N_m}(\ket\psi)^2
\end{aligned}\end{equation}
which means that the value of $\TTT{m-1}_\SEP$ will be given by a corresponding gauge function $\Gamma_{\V \cup \N_m}(\ket\psi)^2$ defined at the level of the underlying Hilbert space, instead of the whole space of Hermitian operators. Since $\V$ and $\N_m$ are both compact sets, by standard results in convex analysis (see e.g.~\cite{rockafellar_1970}, 16.4.1 and 15.1.2), this gauge can be obtained as
\begin{equation}\begin{aligned}
  \Gamma_{\V \cup \N_m}(\ket\psi) = \min_{\ket\psi = \ket{x} + \ket{y}} \Gamma_\V (\ket{x}) + \Gamma_{\N_m} (\ket{y}).
\end{aligned}\end{equation}

Now, for any vector $\ket{x}$ we have $\Gamma_{\N_m} (\ket{x}) = \sqrt{m} \lnorm{x}{2} = \sqrt{m} \lnorm{\ket{\xi_x}}{2}$ and it is known that $\Gamma_\V (\ket{x})$ can be computed as $\lnorm{\ket{\xi_x}}{1}$ (see e.g.~\cite{rudolph_2005,johnston_2015}). By optimising over vectors $\ket{x},\ket{y}$ in the Schmidt basis of $\ket\psi$ only, the problem reduces to the $m$-distillation norm of the Schmidt vector $\ket{\xi_\psi}$, and we thus have $\TTT{m-1}_\SEP(\psi) \leq \mnorm{\ket{\xi_\psi}}^2 - 1$.

To show the opposite inequality, we use the fact that $(\V \cup \N_m)^\circ = \V^\circ \cap \N_m^\circ$ to write the gauge $\Gamma_{\V \cup \N_m}$ in its dual form as
\begin{equation}\begin{aligned}
  &\TTT{m-1}_\SEP(\psi) + 1 = \\ & \quad \max \lset \cbraket{\psi|x}^2 \bar \Gamma^\circ_{\N_m} (\ket x) \leq 1,\; \Gamma^\circ_{\V}(\ket x) \leq 1 \rset
\end{aligned}\end{equation}
where we have, for any $\ket{x}$, $\Gamma^\circ_{\V}(\ket x) = \lnorm{\ket{\xi_x}}{\infty}$~\cite{shimony_1995}. By optimising over all vectors $\ket{x}$ in the Schmidt basis of $\ket\psi$, we recover again the $m$-distillation norm of $\ket{\xi_\psi}$ and the result follows.
\end{proof}

From the above, we then have the fidelity of distillation under $\SEPP$ of a pure state $\ket\psi$ as $\frac{1}{m}\mnorm{\ket{\xi_\psi}}^2$. Crucially, the $m$-distillation norm can be closely connected with the concept of majorisation, allowing us to relate it to the optimal fidelity of pure-state distillation under LOCC and one-way LOCC ($\OLOCC$), which was previously considered in \cite{vidal_2000-1}. We will now rederive the exact expression for the fidelity of pure-state distillation under LOCC in terms of the $m$-distillation norm, and in particular establish an operational equivalence between all relevant sets of operations in the distillation of entanglement from pure states.

\begin{theorem}\label{thm:all_fid_equal_LOCC}
For any pure state $\ket\psi$, any integer $m \geq 1$, and any set of operations $\O \in \{\OLOCC, \LOCC, \PPTO, \PPTRPPR, \PPTPPR, \SEPP \}$, the fidelity of distillation is given by 
\begin{equation}\begin{aligned}
  &F_\O (\psi,m) =\frac{1}{m} \mnorm{\ket{\xi_\psi}}^2.
\end{aligned}\end{equation}
\end{theorem}
\begin{proof}
We begin by recalling that the $m$-distillation norm of $\ket{\xi_\psi} \coloneqq \left(\alpha_1, \ldots, \alpha_d\right)^T$ can be computed as
\begin{equation}\begin{aligned}
  \mnorm{\ket{\xi_\psi}} = \lnorm{\ket{\alpha^\downarrow_{1:m-k^\star}}}{1} + \sqrt{k^\star} \lnorm{\ket{\alpha^\downarrow_{m-k^\star+1:d}}}{2},
\end{aligned}\end{equation}
with $k^\star \coloneqq \displaystyle\argmin_{1 \leq k \leq m} \frac{\lnorm{\ket{\alpha^\downarrow_{m-k+1:d}}}{2}^2}{k}$. 

Now, it is well-known that the deterministic transformation from $\ket\psi$ to another pure state $\ket\eta$ is possible with (either one-way and two-way) $\LOCCO$ if and only if the Schmidt vector $\ket{\xi_\psi} = \left(\alpha_1, \ldots, \alpha_d\right)^T$ is majorised by the Schmidt vector $\ket{\xi_\eta} \coloneqq \left(\beta_1, \ldots, \beta_d\right)^T$~\cite{nielsen_1999}, that is,
\begin{equation}\begin{aligned}
  \sum_{i=1}^{k} \alpha^2_i \leq \sum_{i=1}^{k} \beta^2_i \quad \forall k \in \{1, \ldots, d\}
\end{aligned}\end{equation}
where we have assumed without loss of generality that the Schmidt coefficients are given in non-increasing order. Let us then define the ansatz
\begin{equation}\begin{aligned}
\ket\eta = \sum_{i=1}^{m-k^\star} \alpha_i \ket{ii} + \sum_{i=m-k^\star+1}^{m} \frac{\lnorm{\ket{\alpha^\downarrow_{m-k^\star+1:d}}}{2}}{\sqrt{k^\star}} \ket{ii}
\end{aligned}\end{equation}
expressed in the Schmidt basis of $\ket{\Psi_m}$, where $k^\star$ is defined as above. To see that the Schmidt coefficients of $\ket\eta$ majorise the ones of $\ket\psi$, let us assume that $k^\star > 1$ (as otherwise the desired relation is trivial) and consider the following chain of equivalent inequalities:
\begin{equation}\begin{aligned}
    \frac{\lnorm{\ket{\alpha^\downarrow_{m-k^\star+1:d}}}{2}^2}{k^\star} &\leq \frac{\lnorm{\ket{\alpha^\downarrow_{m-k^\star+2:d}}}{2}^2}{k^\star-1}\\
    \alpha^2_{m-k^\star+1} + \lnorm{\ket{\alpha^\downarrow_{m-k^\star+2:d}}}{2}^2 &\leq \frac{k^\star}{k^\star-1} \lnorm{\ket{\alpha^\downarrow_{m-k^\star+2:d}}}{2}^2\\
    \alpha^2_{m-k^\star+1} &\leq \frac{1}{k^\star-1} \left(1 - \lnorm{\ket{\alpha^\downarrow_{1:m-k^\star+1}}}{2}^2\right)\\
    k^\star \alpha^2_{m-k^\star+1} &\leq \left(1 - \lnorm{\ket{\alpha^\downarrow_{1:m-k^\star}}}{2}^2\right)\\
    \alpha^2_{m-k^\star+1} &\leq \frac{\lnorm{\ket{\alpha^\downarrow_{m-k^\star+1:d}}}{2}^2}{k^\star}
\end{aligned}\end{equation}
where the first line follows by definition of $k^\star$, and in the third and fifth lines we have used the fact that $\ket\psi$ is a normalised pure state.
Thus, we have
\begin{equation}\begin{aligned}
  F_\OLOCC(\psi,m) &\geq \cbraket{\Psi_m|\eta}^2\\
  &= \frac1m \left(\lnorm{\alpha^\downarrow_{1:m-k^\star}}{1} + \sqrt{k^\star} \lnorm{\alpha^\downarrow_{m-k^\star+1:d}}{2}\right)^2\\
  &= \frac{1}{m} \mnorm{\ket{\xi_\psi}}^2.
\end{aligned}\end{equation}

On the other hand, since $\OLOCC \subseteq \SEPP$, we have
\begin{equation}\begin{aligned}
  F_\OLOCC(\psi,m) &\leq F_\SEPP(\psi, m) \\ &= \frac{1}{m} \left(\TTT{m-1}_\SEP(\psi) + 1\right) \\&= \frac{1}{m} \mnorm{\ket{\xi_\psi}}^2
\end{aligned}\end{equation}
by Thm.~\ref{thm:Tm_pure_states}, which concludes the proof.
\end{proof}

The exact correspondence is rather surprising. The operations PPT and SEPP are known to be much more powerful than LOCC in general --- indeed, $\SEPP$ exhibit no bound entanglement whatsoever~\cite{brandao_2010}, and even in the manipulation of pure states PPT operations can, for instance, arbitrarily increase the Schmidt rank (number of Schmidt coefficients) of a pure state~\cite{ishizaka_2005,matthews_2008}, which cannot increase whatsoever under $\LOCCO$ or $\SEPO$~\cite{lo_2001}. The result then shows that even such large sets cannot outperform one-way LOCC in entanglement distillation from pure states, even in the one-shot setting.

Using the above expression, we can furthermore show that the computation of $E^{(1),\ve}_{d,\O}$ belongs to a class of efficiently solvable optimisation problems known as quadratically-constrained linear programs~\cite{lobo_1998}.
\begin{corollary}
For any set of operations $\O \in \{\OLOCC, \LOCC, \PPTO, \PPTRPPR, \PPTPPR, \SEPP \}$, the one-shot distillable entanglement of a pure state can be expressed exactly as the optimal value of the convex quadratically-constrained linear program
\begin{align}
 & E^{(1),\ve}_{d,\O} (\psi)\nonumber\\
 & = \Big\lfloor - \log \min \Big\{ \norm{\ket\omega}{\ell_\infty}^2 \;\Big|\; \braket{\xi_\psi|\omega} \geq \sqrt{1-\ve},\\
  &\hphantom{= \Big\lfloor - \log \min \Big\{ \norm{\ket\omega}{\ell_\infty}^2 \;\Big|} \norm{\ket\omega}{\ell_2} \leq 1,\; \ket\omega \in \RR^d_+ \Big\} \Big\rfloor_{\log}.\nonumber
\end{align}
\end{corollary}
\begin{proof}
The dual form of the $m$-distillation norm, which we recall here as
\begin{equation}\begin{aligned}
   \mnorm{\ket{x}} = \max \lset \cbraket{x|w} \bar \lnorm{\ket{w}}{\infty} \leq 1,\; \lnorm{\ket{w}}{2} \leq \sqrt{m} \rset,
\end{aligned}\end{equation}
gives
\begin{equation}\begin{aligned}
  E^{(1),\ve}_{d,\O} (\psi) &= \log \max \lsetr m \in \mathbb N \barr \frac{1}{m} \mnorm{\ket{\xi_\psi}}^2 \geq 1- \ve \rsetr\\
  &= \Big\lfloor - \log \min \Big\{ \norm{\ket\omega}{\ell_\infty}^2 \;\Big|\; \cbraket{\xi_\psi|\omega}^2 \geq 1-\ve,\\
  &\hphantom{= \Big\lfloor - \log \min \Big\{ \norm{\ket\omega}{\ell_\infty}^2 \;\Big|} \norm{\ket\omega}{\ell_2} \leq 1 \Big\} \Big\rfloor_{\log}
\end{aligned}\end{equation}
and we conclude by noting that it suffices to optimise over vectors with non-negative coefficients since $\ket{\xi_\psi}$ is also non-negative.
\end{proof}
The above result can be compared with the bounds obtained for LOCC distillable entanglement in~\cite{buscemi_2010-1,buscemi_2013}, and in fact we have tightened the bounds to an exact expression for the one-shot distillable entanglement:
\begin{equation}\begin{aligned}
  E^{(1),\ve}_{d,\O} (\psi) \!=\! \logfloor{ - \log \min \Big\{\left. \norm{\rho_B}{\!\infty} \,\right|\, F\!\left(\rho_B, \Tr_{A}(\psi)\right) \!\geq\! 1-\ve \Big\} }\!\!.
\end{aligned}\end{equation}

The Theorem also leads to an interesting characterisation of the $m$-distillation norm in two different ways. Notice that the proof of Thm.~\ref{thm:Tm_pure_states} in fact shows that the $m$-distillation norm of the Schmidt vector $\ket{\xi_x}$ of a vector $\ket{x}$ can be equivalently written as a norm at the level of the vector $\ket{x}$ itself:
\begin{equation}\begin{aligned}\label{eq:mnorm_locc_relation}
   &\mnorm{\ket{\xi_{x}}} = \min_{\ket{x}=\ket{y}+\ket{z}} \lnorm{\ket{\xi_y}}{1} + \sqrt{m} \lnorm{\ket{\xi_z}}{2}\\
  &= \max \lset \cbraket{x|w} \bar \lnorm{\ket{\xi_w}}{\infty} \leq 1,\; \lnorm{\ket{\xi_w}}{2} \leq \sqrt{m} \rset.
\end{aligned}\end{equation}
Writing $\ket\psi \toLOCC \ket\eta$ to denote that the pure state transformation is possible with $\LOCCO$, we then have
\begin{equation}\begin{aligned}
  \mnorm{\ket{\xi_\psi}} &= \sqrt{m}\,  \max \lsetr \cbraket{\Psi_m|\eta} \barr \ket\psi \xrightarrow{\LOCCO} \ket\eta \rsetr\\
  &= \sqrt{m}\, \max \lsetr \cbraket{\psi|\eta} \barr \ket\eta \xrightarrow{\LOCCO} \ket{\Psi_m} \rsetr,\\
\end{aligned}\end{equation}
where the maximisation is over normalised state vectors $\ket\eta$, and the second line is precisely Eq.~\ref{eq:mnorm_locc_relation}.

As a straightforward corollary of the results above, we can establish the value of the quantifiers $G^{(m)}_\S$ for several sets other than $\SEP$.
\begin{corollary}For any $1 \leq m \leq d$, any pure state $\ket\psi$, and any $\S \in \{\SEP,\PPT,\PPTP,\PPTR\}$ we have
\begin{equation}\begin{aligned}
  G^{(m)}_\S (\proj\psi) = \frac{1}{m} \mnorm{\ket{\xi_\psi}}^2
\end{aligned}\end{equation}
\end{corollary}

Going beyond pure states, combined with Prop.~\ref{prop:rob_for_d-1} the expression for $F_\SEPP(\rho,d)$ gives a direct operational interpretation to the generalised robustness of entanglement $R_\SEP^\DD$ by showing that
\begin{equation}\begin{aligned}
  R_\SEP^\DD(\rho) = d \FF{\SEPP}(\rho,d) - 1
\end{aligned}\end{equation}
for any state $\rho$, and in fact by Thm.~\ref{thm:ppt_dist} also for the class $\PPTPPR$ we have the relation $R_\PPTP^\DD(\rho) = d F_{\PPTPPR}(\rho,d) - 1$. This complements the known operational applications of this quantity~\cite{brandao_2011,takagi_2019-2,bae_2019}. Note also that all of the other measures in the family $\TTT{m}_\SEP$ and $\TTT{m}_\PPTP$, including the modified trace distance quantifiers, are given similar interpretations.

An important use of the fidelity of distillation $F_{\LOCCO}(\rho,m)$ in the particular case $m=d$ is as the fidelity of teleportation, that is, the best average fidelity one can achieve in the task of quantum teleportation by employing an LOCC protocol on the state $\rho$~\cite{horodecki_1999-1}. Notably, in~\cite{verstraete_2003-2} it was then shown that for $d_A = d_B = 2$, we have
  $F_{\LOCCO} (\rho, d) = F_{\PPTO} (\rho,d)$
showing that even PPT protocols (or SEPP protocols, by Lemma~\ref{lemma:sepp_fid}) cannot enhance the fidelity of teleportation of the given state. By Thm.~\ref{thm:all_fid_equal_LOCC}, we know that this relation extends to all pure states in all dimensions; that is, $F_{\LOCCO} (\psi, d) = F_{\PPTO} (\psi,d) = F_{\SEPP} (\psi,d)$.

\begin{remark}In~\cite{brandao_2005} (Prop. 9) it was claimed that the asymptotic distillable entanglement $E^\infty_{d,\LOCC}(\rho)$ of any state is upper bounded by $\log \left(\TTT{m}_\SEP(\rho) + 1\right)$ for any $m \geq 1$. This is clearly not true, as $\TTT{m}_\SEP(\rho) \leq m \; \forall \rho$ and the distillable entanglement obeys no such restriction. We have seen, however, that the quantifiers $\TTT{m}_\SEP(\rho)$ characterise exactly the \textit{fidelity} of distillation.\end{remark}

As a side note, noticing the similarity between the distillation under $\PPTP$-preserving, Rains-preserving, and $\SEP$-preserving operations, it might appear that the hypothesis testing relative entropy $D^\ve_H$ in general quantifies the rate of distillation under a set of operations which is defined to preserve a given set of operators. This claim is supported by recent independent results concerned with distillation in a class of general quantum resources~\cite{liu_2019}, but it does not hold in full generality as the distillation under $\PPTPR$ operations shows (see Thm.~\ref{thm:ppt_dist}), and indeed also distillation in the resource theory of coherence~\cite{regula_2017} is a counterexample.

\subsection{Isotropic states}\label{sec:iso}

Consider $d=d_A=d_B$ and define the isotropic states as
\begin{equation}\begin{aligned}
  \rho_f = f \Psi_d + \frac{1 - f}{d^2 - 1} \left(\id - \Psi_d\right)
\end{aligned}\end{equation}
with $0 \leq f \leq 1$. This class of states is particularly useful due to its strong symmetry, allowing for a much easier evaluation of their entanglement properties~\cite{horodecki_1999}.
We then have the following result, showing the operational equivalence of all sets of channels from $\SEPO$ to $\PPTPPR$ and $\SEPP$ in distilling entanglement from isotropic states, and extending the known characterisation of isotropic state distillation under PPT operations considered in~\cite{rains_2001}.

\begin{theorem}\label{thm:isotropic}
For any isotropic state $\rho_f$ and any $\O \in \{\SEPO, \PPTO, \PPTRPPR, \PPTPPR, \SEPP\}$, it holds that
\begin{equation}\begin{aligned}
  F_\O(\rho_f,m)  = \begin{cases} \frac{1}{m} & f \leq \frac1d\\ \frac{df - 1}{d-1} + \frac{d(1-f)}{m(d-1)} & f \geq \frac1d.\end{cases}
\end{aligned}\end{equation}
\end{theorem}
\begin{proof}
Take $1 \leq m \leq d$. If $f \leq \frac{1}{d}$, then $\rho_f \in \SEP$, and we have $F_\SEPO(\rho_f,m) = F_\PPTO(\rho_f,m) = F_\SEPP(\rho_f,m) = \frac{1}{m}$; we will therefore assume that $f \geq \frac1d$ in the sequel.

Recall by our previous arguments that, due to twirling, we can limit ourselves to considering operations of the form
\begin{equation}\begin{aligned}
  \Lambda_{W}(Z) =  \< Z, W\> \Psi_m + \frac{\< Z, \id - W \>}{m^2-1} \left(\id - \Psi_m\right),
\end{aligned}\end{equation}
and the fidelity of distillation under a set $\O$ is then given by
\begin{equation}\begin{aligned}
  F_\O(\rho_f,m) = \max \lset \< \rho_f, W \> \bar \Lambda_W \in \O \rset.
\end{aligned}\end{equation}
Here, noting the invariance of $\rho_f$ under twirling, we can twirl once more; in particular, $\<\rho_f, W \> = \< \T(\rho_f), \T(W) \> = \< \rho_f, \T(W) \> \; \forall W$, so we can again limit the considered $W$ to be of the form
\begin{equation}\begin{aligned}
  W = \alpha \Psi_d + \beta \id,
\end{aligned}\end{equation}
where $0 \leq \beta \leq 1$, $0 \leq \alpha + \beta \leq 1$. The Choi operator of the corresponding map $\Lambda_W$ is then of the form
\begin{equation}\begin{aligned}
  J_{\Lambda_W} =& \frac{\alpha m^2}{m^2 - 1} \Psi_{dm} + \frac{\beta m^2 - 1}{m^2-1} \id_{AB} \otimes \Psi_m \\
  &- \frac{\alpha}{m^2-1} \Psi_d \otimes \id_{A'B'} + \frac{1-\beta}{m^2-1} \id_{ABA'B'}
\end{aligned}\end{equation}
where $d_{A'} = d_{B'} = m$. By~\cite[Thm. 6]{lami_2016-2}, $J_{\Lambda_W} \in \SEP(AA':BB')$ if and only if the following conditions are all satisfied:
\begin{equation}\begin{aligned}\label{eq:sep_ineq}
  & d - m \alpha + d m \beta \geq 0,\\
  & d + m \alpha - d m \beta \geq 0,\\
  & d - m \alpha - d m \beta \geq 0,\\
  & m \alpha + d^2 m \beta - d (\alpha+\beta) \geq 0.
\end{aligned}\end{equation}
Let us choose
\begin{equation}\begin{aligned}
  W = \frac{d(m-1)}{(d-1)m} \Psi_d + \frac{d-m}{(d-1{})m}\id
\end{aligned}\end{equation}
for which the inequalities \eqref{eq:sep_ineq} can be readily verified to hold. This gives
\begin{equation}\begin{aligned}
  F_\SEPO(\rho_f,m) &\geq \< \rho_f, W \>\\
  &= \frac{df - 1}{d-1} + \frac{d(1-f)}{m(d-1)}.
\end{aligned}\end{equation}
On the other hand, take
\begin{equation}\begin{aligned}
  X = \frac{d(1-f)}{d^2-1}\left(\Psi_d + \frac{1}{d}\id\right).
\end{aligned}\end{equation}
It is known that $\Psi_d + \frac{1}{d} \id \in \SEP\*\*$~\cite{vidal_1999}, which gives
\begin{equation}\begin{aligned}
  F_\SEPP(\rho_f,m) &\leq \Tr (\rho_f - X)_+ + \frac{1}{m} \Tr X\\
  &= \frac{df - 1}{d-1} + \frac{d(1-f)}{m(d-1)}
\end{aligned}\end{equation}
where we used Lemma~\ref{lemma:sepp_fid} together with the dual form of $\GG{m}_\SEP$ from Prop.~\ref{prop:G_primal}.
\end{proof}

We remark that the above also gives a general way of lower-bounding the fidelity of distillation under separable operations of any state with a simple linear program, tight for all isotropic states:
\begin{equation}\begin{aligned}
  &F_\SEPO(\rho,m) \geq\\
  & \max \lset \alpha \!\< \rho, \Psi_d \>\! +\! \beta \bar 0 \leq \beta \leq 1,\; 0 \leq \alpha+\beta \leq 1,\; \text{Eq. \eqref{eq:sep_ineq}} \rset.
\end{aligned}\end{equation}
Unsurprisingly, however, a numerical investigation reveals this bound to be rather ineffective beyond the set of isotropic states.

Note also that a general investigation of one- and multi-shot entanglement distillation from isotropic states under PPT operations as a linear program has been explored in~\cite{rains_2001,fang_2019}.

\subsection{Maximally correlated states}\label{sec:maxcorr}

Let us consider a bipartite system with $d = d_A = d_B$. A \textit{maximally correlated state} is any state of the form $\rho_\mc = \sum_{i,j} \rho_{ij} \ket{ii}\!\bra{jj}$ for some local orthonormal bases $\{\ket{i}\}$~\cite{rains_1999-1}. The name for this class of states comes from the fact is that the two parties are guaranteed to obtain the same measurement results for any measurement in their local basis $\{\ket{i}\}$.

Notice that any maximally correlated state has a corresponding single-party state $\widetilde{\rho_\mc} \coloneqq \sum_{i,j} \rho_{ij} \ketbra{i}{j}$ with the same coefficients in an orthonormal basis $\{\ket{i}\}$. This led to comparisons between the manipulation of maximally correlated states and the resource theory of coherence, which studies the properties of superposition as a quantum resource~\cite{baumgratz_2014,streltsov_2017}. In particular, it has been conjectured in several works that the resource theory of coherence is equivalent to the resource theory of entanglement restricted to maximally correlated states~\cite{winter_2016,chitambar_2016}. Although this conjecture is still unsolved in full generality, we know that many operational quantifiers such as the entanglement of formation, relative entropy of entanglement (and other R\'{e}nyi entropy--based entanglement measures), and asymptotic distillable entanglement can be evaluated on maximally correlated states by quantifying the corresponding coherence quantifiers, typically significantly simpler to evaluate and satisfying useful properties such as additivity~\cite{winter_2016,zhu_2017}. Furthermore, an operational equivalence between transformations acting on $\widetilde{\rho_\mc}$ and LOCC operations acting on $\rho_\mc$ has been suggested, although so far this conjecture has been shown only in specific cases~\cite{winter_2016,chitambar_2016}.

To obtain a result allowing us to quantify the one-shot distillable entanglement of maximally correlated states, we will fix a choice of basis $\{\ket{i}\}_{i=1}^{d}$ for the single-party state $\widetilde{\rho_\mc}$, and use $\I \coloneqq \conv \{ \proj{i} \}_i$ to denote the set of all incoherent (diagonal) states in this basis. Furthermore, we define the subset of separable states $\SEP_\mc \coloneqq \conv \{ \proj{ii} \}_i$ where $\{\ket{ii}\}_{i=1}^{d}$ is the maximally correlated basis of the state $\rho_\mc$. We then get the following.

\begin{theorem}\label{thm:max_corr_ppt}
For any maximally correlated state, any $m \geq 1$, and any choice of operations $\O \in \{ \PPTO, \PPTRPPR, \PPTPPR, \SEPP \}$ it holds that
\begin{equation}\begin{aligned}
  F_{\O}(\rho_\mc,m) = \GG{m}_\I (\widetilde{\rho_\mc}).
\end{aligned}\end{equation}
\end{theorem}
\begin{proof}
Using Lemma~\ref{lemma:sepp_fid} together with Prop.~\ref{prop:G_primal}, we have
\begin{equation}\begin{aligned}
  F_{\PPTO}(\rho_\mc,m) &\leq F_{\SEPP}(\rho_\mc,m)\\
   &=  \min_{X \in \SEP\*\*} \Tr(\rho_\mc - X)_+ + \frac1m\Tr(X)\\
  &\leq  \min_{X \in \SEPmcdd} \Tr(\rho_\mc - X)_+ + \frac1m\Tr(X)\\
  &= \min_{X \in \I\*\*} \Tr(\widetilde{\rho_\mc} - X)_+ + \frac1m \Tr(X)\\
  &= \GG{m}_\I(\widetilde{\rho_\mc}).
\end{aligned}\end{equation}
On the other hand, let $W^\star = \sum_{i,j} W_{ij} \ket{i}\bra{j}$ be the optimal solution to the dual problem of
\begin{equation}\begin{aligned}
  \GG{m}_\I (\widetilde{\rho_\mc}) = \max \lset \< \widetilde{\rho_\mc}, W \> \bar 0 \mleq W \mleq \id,\; W \in \frac{1}{m} \I^\circ \rset.
\end{aligned}\end{equation}
Notice that $W^\star \in \frac{1}{m} \I^\circ$ is equivalent to $\max_i W_{ii} \leq \frac1m$. Consider then the matrix $W_\mc \coloneqq \sum_{i,j} W_{ij} \ket{ii}\bra{jj}$, defined in the basis of $\rho_\mc$. The eigenvalues of $W^{T_B}_\mc$ can be straightforwardly verified to be $\{ W_{ii}, \;\pm|W_{ij}| \}_{i,j=1}^d$. Positivity of $W^\star$ imposes that 
\begin{equation}\begin{aligned}
  \max_{i,j} |W_{ij}| \leq \max_i W_{ii}
\end{aligned}\end{equation}
from which it follows that
\begin{equation}\begin{aligned}
  \Gamma_{\PPTR^\circ}(W_\mc) = \norm{W^{T_B}_\mc}{\infty} = \max_i W_{ii} \leq \frac1m.
\end{aligned}\end{equation}
This means that $W_\mc \in \frac{1}{m}\PPTRc$, and so
\begin{equation}\begin{aligned}
  &F_{\PPTO}(\rho_\mc,m) \\
  &= \GG{m}_\PPTR (\rho_\mc) \\
  &=  \max \lsetr \< \rho_\mc, W \> \barr 0 \mleq W \mleq \id,\; W \in \frac{1}{m} \PPTRc \rsetr\\
  &\geq  \< \rho_\mc, W_\mc \>\\
  &=  \GG{m}_\I (\widetilde{\rho_\mc}).
\end{aligned}\end{equation}
\end{proof}
Notice that $ \GG{m}_\I (\widetilde{\rho_\mc}) = \frac{1}{m} \left( \TTT{m-1}_\I (\widetilde{\rho_\mc}) + 1 \right)$, where $\TTT{m-1}_\I$ have been considered as coherence measures in~\cite{regula_2017}. Further, using Thm.~\ref{thm:hyp_entropy} we have that
\begin{equation}\begin{aligned}
  E^{(1),\ve}_{d,\O} (\rho_\mc ) &= \logfloor{ \min_{\sigma \in \I} D^{\ve}_H(\widetilde{\rho_\mc}\|\sigma) }\\
  &= \logfloor{ - \log \min_{\substack{\< \widetilde{\rho_\mc}, W \> \geq 1 - \ve\\0 \mleq W \mleq \id}}\, \norm{\Delta(W)}{\infty} }
\end{aligned}\end{equation}
where $\Delta(\cdot) = \sum_i \proj{i} \cdot \proj{i}$ is the completely dephasing map. We stress that these optimisation problems are all efficiently computable as simple semidefinite programs~\cite{vandenberghe_1996}, facilitating an efficient quantification of the fidelity as well as rates of one-shot distillation of all maximally correlated states.

Interestingly, in contrast to many other results which show an exact equality between operational quantities in the resource theory of coherence and the resource theory of entanglement of maximally correlated states, our result above shows a slight discrepancy between the two resources: in particular, it is not difficult to find numerical examples of states such that
\begin{equation}\begin{aligned}
  F_\SEPP (\rho_\mc, m) = \GG{m}_\I (\rho) > \GG{m}_\J (\rho) =  F_\MIO (\rho, m)
\end{aligned}\end{equation}
$\forall\, m < d$, where MIO denotes the class of maximally incoherent operations in the resource theory of coherence, defined to be channels $\Lambda$ such that $\sigma \in \I \Rightarrow \Lambda(\sigma) \in \I$, and $\J$ is the set of all unit-trace diagonal Hermitian operators (see~\cite{regula_2017} for the rightmost equality). Therefore, the one-shot distillable entanglement of a maximally correlated state under the largest set of free operations in the resource theory of entanglement (SEPP) can be strictly larger than the distillable coherence of the corresponding single-partite state under the largest set of free operations in the resource theory of coherence (MIO). This shows in particular that, in the distillation of entanglement from maximally correlated states under SEPP, it is not sufficient to consider operations whose output remains in the maximally correlated subspace --- indeed, if this were the case, any such operation could always be mapped to a corresponding MIO operation, and the fidelities $F_\MIO(\rho,m)$ and $F_\SEPP(\rho_\mc,m)$ would be equal. This also motivates a rather curious conjecture that, should there exist a smaller class of operations for which it suffices to consider only maximally correlated output states, then it is plausible that $F_\LOCC(\rho_\mc,m) \leq F_\MIO(\rho,m) < F_\PPTO(\rho_\mc,m)$ for general maximally correlated states. This could be surprising, as it is known that the gap between $\LOCC$ and $\PPTO$ distillation of $\rho_\mc$ disappears at the asymptotic level~\cite{devetak_2005,hiroshima_2004} or even when considering the second-order non-asymptotic expansion of the rate of distillation~\cite{fang_2019}.


\subsection{Assisted distillation}\label{sec:assist}

The setting of (environment-)assisted distillation of entanglement, considered first in~\cite{divincenzo_1999,smolin_2005}, has been studied in the non-asymptotic regime in~\cite{buscemi_2013}. It is based on a scenario in which the two parties $A$ and $B$ are assisted by a third party $C$ who holds a purifying state of the system $\rho_{AB}$, i.e. such that the joint state is $\psi_{ABC}$, and aims to increase the entanglement distillable from $\rho_{AB}$ by performing a measurement on their local system $C$ and communicating its result classically to $A$ and $B$. A particular property of this setting is that the optimal protocol always involves a rank-1 measurement on subsystem $C$~\cite{buscemi_2013}, giving parties $A$ and $B$ access to arbitrary pure-state decompositions of the system $\rho_{AB}$. Specifically, the best achievable rate of distillation is given by
\begin{equation}
E_{A,\O}^{(1),\ve}(\rho) \coloneqq \log \max \lset m \in \mathbb N \bar F_{A,\O}(\rho,m) \geq 1- \ve \rset
\label{one shot asst distillable ent}
\end{equation}
where the fidelity of assisted distillation is the best average fidelity optimised over all decompositions, i.e.
  \begin{align}
       &F_{A,\O}(\rho,m) \coloneqq\\
       &= \max \left\{ \vphantom{\< \sum_i p_i \>} \right. \< \sum_i p_i \Lambda_i(\proj{\psi_i}), \Psi_m \> \;\left|\vphantom{\< \sum_i p_i \>} \right.\; \rho = \sum_i p_i \proj{\psi_i},\nonumber\\
       & \hphantom{= \max \Big\{ \< \sum_i p_i \Lambda_i(\proj{\psi_i}), \Psi_m \> \;\Big|\;\; } \Lambda_i \in \O \; \forall i \left\}\vphantom{\< \sum_i p_i \>} \right..\nonumber
  \end{align}
This in particular means that, having obtained the measurement result from party $C$, the distillation is performed from a \textit{pure} state --- therefore, employing our results in Thm.~\ref{thm:all_fid_equal_LOCC}, we immediately obtain the result that the rate of assisted entanglement disitillation is the same under all sets of operations from $\OLOCC$ up to $\PPTPPR$ and $\SEPP$.

Additionally, the proofs of the main results of~\cite{buscemi_2013} can be significantly simplified by employing the formalism introduced in this work, in fact strengthening the one-shot characterisation of Thms. 1 and 2 of~\cite{buscemi_2013} and tightening the bounds derived therein. In particular, our pure state-results in Thm.~\ref{thm:all_fid_equal_LOCC} allow us to straightforwardly obtain the following.
\begin{theorem}
For any $\O \in \{\OLOCC$, $\LOCC$, $\PPTO$, $\PPTPR$, $\PPTPPR$, $\SEPP \}$, the fidelity and one-shot rate of assisted distillation of any state are given by
\begin{align}
  &F_{A,\O}(\rho,m) = \max \lsetr F(\rho, \omega) \barr \omega \in \M_m \rsetr,\\
  &E_{A,\O}^{(1),\ve}(\rho) = \label{eq:assisted_rate}\\
  &\quad \logfloor{ -\log \min \lset \vartheta(\omega) \bar \omega \in \DD,\; F(\rho,\omega) \geq 1-\ve \rset },\nonumber
\end{align}
where
\begin{equation}\begin{aligned}
  \M_m \coloneqq \conv \lsetr \proj{\phi} \barr \lnorm{\ket{\phi}}{2} = 1,\; \lnorm{\ket{\xi_\phi}}{\infty} \leq \frac{1}{\sqrt{m}} \rsetr
\end{aligned}\end{equation}
and
\begin{equation}\begin{aligned}
  \vartheta(\omega) \coloneqq \min \lset \max_i \lnorm{\ket{\xi_{\psi_i}}}{\infty}^2 \bar \omega = \sum_i p_i \proj{\psi_i} \rset.
\end{aligned}\end{equation}
\end{theorem}
\begin{proof}
The derivation follows the approach taken for quantum coherence in~\cite{regula_2018-1}. We begin by writing the fidelity of assisted distillation as
\begin{equation}\begin{aligned}
  F_{A,\O}(\rho,m) = \max \!\lset \sum_i p_i F_{\O}(\psi_i,m) \bar \rho = \sum_i p_i \proj{\psi_i} \rset
\end{aligned}\end{equation}
with $F_\O$ denoting the fidelity of distillation as before.
Notice now that $F_{\O}(\psi_i,m) = \frac{1}{m} \mnorm{\ket{\xi_{\psi_i}}}^2$ can be written as $F_\O(\psi_i,m) = \max_{\omega \in \M_m} F(\psi_i, \omega)$, where we employed the dual characterisation of the $m$-distillation norm of the Schmidt vector. Since $\M_m$ is defined as the convex hull of rank-one projectors, we can now use the result of Streltsov et al.~\cite{streltsov_2010} (see also~\cite{regula_2018-1}) to obtain
\begin{equation}\begin{aligned}
  &F_{A,\O}(\rho,m)\\
  &= \max \lset \sum_i p_i \max_{\omega_i \in \M_m} F(\psi_i,\omega_i) \bar \rho = \sum_i p_i \proj{\psi_i} \rset\\
  &= \max_{\omega \in \M_m} F(\rho,\omega)
\end{aligned}\end{equation}
as required. The quantity $\vartheta$ is simply a function defined so that any state $\omega$ satisfies $\omega \in \M_m \iff \vartheta(\omega) \leq \frac{1}{m}$, allowing us to obtain
\begin{equation}\begin{aligned}
  &E_{A,\O}^{(1),\ve}(\rho) \\
  &\coloneqq \log \max \lsetr m \in \mathbb N \barr \max_{\omega \in \M_m} F(\rho, \omega) \geq 1- \ve \rsetr\\
  &\,= \log \max \Big\{ m \in \mathbb N \;\Big|\; \omega \in \DD,\; \vartheta(\omega) \leq \frac1m,\\
  & \hphantom{\;= \log \max \Big\{ m \in \mathbb N \;\Big|} F(\rho, \omega) \geq 1- \ve \Big\}
\end{aligned}\end{equation}
and thus completing the proof.
\end{proof}


\subsection{Zero-error distillation}\label{sec:zero}

Taking $\ve=0$ in the task of one-shot distillation corresponds to the problem of characterising the exact transformation $\rho \to \Psi_m$ with a given class of free operations. One is then interested in understanding not only the one-shot capabilities in such a task, but also the asymptotically achievable rate
\begin{equation}\begin{aligned}
  E_{d,\O}^{\infty, 0}(\rho) \coloneqq \limsup_{n \to \infty} \frac1n E_{d,\O}^{(1),0}(\rho^{\otimes n}).
\end{aligned}\end{equation}

To apply our methods in this setting, let us focus on the classes of operations for which we have shown that $F_\O(\rho,m) = \GG{m}_\S(\rho)$ for some set $\S$; recall from our previous results that $F_\SEPP(\cdot,m) = \GG{m}_\SEP$, $F_{\PPTPPR}(\cdot,m) = \GG{m}_\PPTP$, $F_{\PPTRPPR} = \GG{m}_\PPTRP$, and $F_\PPTPR(\cdot,m) = F_\PPTO(\cdot,m)  = \GG{m}_\PPTR$.

\begin{lemma}
Take $\O \in \{ \PPTO, \PPTPR, \PPTRPPR, \PPTPPR, \SEPP\}$ and let $\S$ be the set such that $F_\O(\rho,m) = \GG{m}_\S(\rho)$ for the given class. Then, for any $\rho$, it holds that
\begin{equation}\begin{aligned}\label{eq:zero_error_general}
  E_{d,\O}^{(1),0}(\rho) = \logfloor { - \log \min \lset \Gamma_{\S^\circ}(W) \bar \Pi_\rho \mleq W \mleq \id \rset }
\end{aligned}\end{equation}
where $\Pi_\rho$ is the projector onto the support of $\rho$.
\end{lemma}
\begin{proof}
Using the characterisation in Prop.~\ref{prop:hyp_gauge} and Thm.~\ref{thm:hyp_entropy}, we can write
\begin{equation}\begin{aligned}\label{eq:zero_error}
  &E_{d,\O}^{(1),0}(\rho)\\
  &= \logfloor { - \log \min \lset \Gamma_{\S^\circ}(W) \bar \< \rho, W \> = 1,\; 0\mleq W \mleq \id \rset }.
\end{aligned}\end{equation}
Write $\rho$ in its spectral decomposition as $\rho = \sum_i \lambda_i \proj{\psi_i}$. We then have
\begin{equation}\begin{aligned}
 \sum_i \lambda_i = \Tr \rho  = 1 = \< W, \rho \> = \sum_i \lambda_i \< W, \proj{\psi_i}\>
\end{aligned}\end{equation}
and so $\braket{\psi_i|W|\psi_i} = 1$ for each $i \in \{1, \ldots, r\}$ since $0 \mleq W \mleq \id$. The constraints then imply that every feasible solution will have the form $W = \Pi_\rho + P$ with $0 \mleq P \mleq \id$ and $\supp(P) \subseteq \ker(\rho)$, and in particular $\Pi_\rho \mleq W \mleq \id$. Conversely, every $W$ such that $\Pi_\rho \mleq W \mleq \id$ satisfies $1 \geq \< \rho, W \> \geq 1$ and $0\mleq W \mleq \id$, so the feasible sets of the two problems are equal.
\end{proof}
Note that for any positive semidefinite $W$, from Eq.~\eqref{eq:gauge_dual} we have $\Gamma_{\S^\circ}(W) = \max_{X \in \S} \<X, W \>$. For the case of $\PPTO$ operations, where $\Gamma_{{\PPTR}^\circ}(W) = \norm{W^{T_B}}{\infty}$, the above recovers a result of~\cite{wang_2016}. 

Notice that the above implies that one-shot zero-error distillation is impossible from any full-rank state under any class of free operations, as for $\Pi_\rho = \id$ the only feasible $W$ is $\id$ itself and so we have $E_{d,\O}^{(1),0}(\rho) = \log \left\lfloor \Gamma_{\S^\circ}(\id)^{-1} \right\rfloor = \log 1 = 0$. We will shortly improve this characterisation of zero-error undistillability.

Interestingly, in the case of $\PPTPPR$, $\PPTRPPR$, and $\SEPP$, the set $\S$ consists of positive semidefinite operators, which means that for any $P \in \HH_+$ it holds that
\begin{equation}\begin{aligned}
  \Gamma_{\S^\circ}(\Pi_\rho + P) &= \max_{\sigma \in \S} \< \sigma, \Pi_\rho + P \> \\
  &\geq \max_{\sigma \in \S} \< \sigma, \Pi_\rho \> = \Gamma_{\S^\circ}(\Pi_\rho)
\end{aligned}\end{equation}
and so $\Pi_\rho$ itself will be the optimal solution to the minimisation in Eq. \eqref{eq:zero_error_general}. This gives the following.
\begin{corollary}
For the classes of operations $\O \in \{ \PPTPPR, \PPTRPPR, \SEPP\}$, the one-shot zero-error distillable entanglement is given exactly by
\begin{equation}\begin{aligned}
  E_{d,\O}^{(1),0}(\rho) = \log \left\lfloor \Gamma_{\S^\circ}(\Pi_\rho)^{-1} \right\rfloor.
\end{aligned}\end{equation}
\end{corollary}
Noticing further that $\Pi_{\rho^{\otimes n}} = \Pi_\rho^{\otimes n}$, we can easily see that $\Gamma_{\S^\circ}\left(\Pi_{\rho^{\otimes n}}\right) \geq \Gamma_{\S^\circ}(\Pi_{\rho})^n$ due to the fact that $\sigma \in \S \Rightarrow \sigma^{\otimes n} \in \S$. This gives the relation $E_{d,\O}^{(1),0}(\rho^{\otimes n}) \leq \logfloor{ - n \log \Gamma_{\S^\circ}(\Pi_\rho)}$, and in particular we see that $-\log \Gamma_{\S^\circ}(\Pi_\rho)$ upper bounds the asymptotically achievable zero-error distillable entanglement $E_{d,\O}^{\infty, 0}(\rho)$. Equality does not generally hold since the quantities $\Gamma_{{\SEP^\circ}}$, $\Gamma_{{\PPT^\circ_+}}$ are not multiplicative (a counterexample being any state supported on the antisymmetric subspace~\cite{werner_2002,grudka_2010}). Interestingly, multiplicativity is indeed satisfied for $\Gamma_{\PPTRP^\circ}$ --- this can be seen explicitly by expressing the computation of $\Gamma_{\PPTRP^\circ}$ in its dual form as
\begin{equation}\begin{aligned}
   \Gamma_{\PPTRP^\circ}(\Pi_\rho) &= \min_{Q \mgeq \Pi_\rho} \norm{Q^{T_B}}{\infty},
\end{aligned}\end{equation}
from which it straightforwardly follows that $\Gamma_{\PPTRP^\circ}(\Pi^{\otimes n}_\rho) \leq \Gamma_{\PPTRP^\circ}(\Pi_\rho)^n$. This gives in particular the following.
\begin{corollary}\label{corr:zero_rains}
The asymptotic zero-error distillable entanglement under Rains-preserving operations is given by
\begin{equation}\begin{aligned}
  E_{d,\O}^{\infty, 0}(\rho) = - \log \Gamma_{\PPTRP^\circ}(\Pi_\rho).
\end{aligned}\end{equation}
\end{corollary}
The result therefore ensures the computability of both one-shot and asymptotic zero-error distillable entanglement under $\PPTRPPR$, showing that it constitutes an efficiently computable upper bound for zero-error LOCC distillation. Note that $\Gamma_{\PPTRP^\circ}(\Pi_\rho)$ appeared previously in the works~\cite{wang_2017-3,wang_2017-1} as a bound on entanglement cost and zero-error distillable entanglement. Our result gives this quantity a precise operational meaning, establishing it as a zero-error equivalent of the Rains bound (cf. Sec.~\ref{sec:rains}; see also discussion in~\cite{wang_2017-3}).

Evaluating $\Gamma_{\SEP^\circ}$ is significantly more difficult~\cite{gurvits_2003,harrow_2010}. One can write this quantity more explicitly as~\cite{acin_2003}
\begin{equation}\begin{aligned}
  \Gamma_{\SEP^\circ}(\Pi_\rho) &= \max \lset \norm{\ket{\xi_\psi}}{\ell_\infty}^2 \bar \ket{\psi} \in \supp(\rho) \rset
\end{aligned}\end{equation}
which makes it easy to see that if the support of $\rho$ contains a product state, then no class of free operations can distill any entanglement without error (even asymptotically). By a result of Parthasarathy~\cite{parthasarathy_2004}, if $\rank(\rho) > (d_A - 1)(d_B - 1)$, then $\Gamma_{\SEP^\circ}(\Pi_\rho) = 1$ and so $E_{d,\SEPP}^{(1),0}(\rho) = 0$. In a very similar manner, if $\supp(\rho)$ contains a PPT state, then $E_{d,\PPTPPR}^{(1),0}(\rho) = 0$; this, however, does not give a better universal bound for the rank of $\rho$ which ensures undistillability~\cite{johnston_2013}.

Our results in previous sections can further simplify the characterisation of zero-error distillable entanglement for several classes of states. In particular, any pure state has
\begin{equation}\begin{aligned}
  E_{d,\O}^{(1), 0}(\psi) = \log \left\lfloor \norm{\ket{\xi_\psi}}{\ell_\infty}^{-2} \right\rfloor
\end{aligned}\end{equation}
for any class of operations $\O$ considered in this work (which was already known in the case of LOCC~\cite{lo_2001} and PPT operations~\cite{matthews_2008}), and a maximally correlated state satisfies
\begin{equation}\begin{aligned}
  E_{d,\O}^{(1), 0}(\rho_{\mc}) = \log \left\lfloor \norm{\Delta(\Pi_{\rho_{\mc}})}{\infty}^{-1} \right\rfloor
\end{aligned}\end{equation}
for any $\O \in \{\PPTO,\PPTRPPR,\PPTPPR,\SEPP\}$, where $\Delta$ is the completely dephasing channel (diagonal map) in the maximally correlated basis. One can furthermore notice that in both of the above cases the quantity $\Gamma_{\SEP^\circ}(\Pi_\rho)$ is multiplicative, which means that in the asymptotic limit we have $E_{d,\O}^{\infty, 0}(\psi) = - \log \norm{\ket{\xi_\psi}}{\ell_\infty}^{2}$ and $E_{d,\O}^{\infty, 0}(\rho_{\mc}) = - \log \norm{\Delta(\Pi_{\rho_{\mc}})}{\infty}$.

Finally, we remark that the quantity $\Gamma_{\SEP^\circ}$, often encountered under the name $h_{\SEP}$, has found a plethora of uses beyond the resource theory of entanglement --- in particular, in the theory of quantum Merlin-Arthur games~\cite{harrow_2013} as well as in characterising the maximum output norms of quantum channels~\cite{werner_2002,harrow_2013}. Indeed, the non-multiplicativity of $\Gamma_{\SEP^\circ}$ is equivalent to the non-multiplicativity of the norm $\norm{\Lambda}{1\to\infty} \coloneqq \max \lset \norm{\Lambda(\rho)}{\infty} \bar \rho \in \DD\rset$ of a channel $\Lambda$; specifically, if $\Lambda$ takes operators on a Hilbert space $\H_{\text{\rm in}}$ to operators on Hilbert space $\H_{\text{\rm out}}$ and $V: \H_{\text{\rm in}} \to \H_{\text{\rm out}} \otimes \H_R$ is an isometry such that $\Lambda(\cdot) = \Tr_{R} V\cdot V^\dagger$ for some auxiliary Hilbert space $\H_R$, then $\norm{\Lambda}{1\to\infty} = \Gamma_{\SEP^\circ}(VV^\dagger)$~\cite{harrow_2013}. This interpretation provides an understanding of the cases in which $\Gamma_{\SEP^\circ}(VV^\dagger)$ is multiplicative: these are the cases in which the protocol $\Lambda$ obeys so-called perfect parallel repetition~\cite{harrow_2013}. It is furthermore known that, although not multiplicative, the quantity $\Gamma_{\SEP^\circ}$ obeys a form of weaker multiplicativity relations~\cite{montanaro_2013,lancien_2017}.
In the context of entanglement distillation we can see that additivity, in the sense that $E_{d,\SEPP}^{\infty, 0}(\rho) = - \log \Gamma_{\SEP^\circ}(\Pi_\rho)$, holds when the optimal operation $\Lambda \in \SEPP$ which distills entanglement from $\rho$ satisfies $\Lambda^{\otimes n} \in \SEPP$ for any $n$. We stress that an additive lower bound on $E_{d,\SEPP}^{\infty, 0}(\rho)$, and therefore also an upper bound on the regularisation of $\norm{\cdot}{1\to\infty}$, is given by Corr.~\ref{corr:zero_rains}.


\section{Discussion}

\noindent The contribution of our work is twofold. 

First, we established a comprehensive set of theoretical tools for the study of entanglement distillation. Employing a general framework based on convex analysis, we were able to relate many operational quantities to convex optimisation problems which can be efficiently characterised, in particular allowing for a significant simplification of the optimisation in many relevant cases. Our results revealed general connections between entanglement monotones $\GG{m}_\Q$ and the hypothesis testing relative entropy $D^\ve_H$, uncovering the fundamental role that both of the quantities play in the task of one-shot entanglement distillation.

Second, the methods found immediate operational applications in characterising the capabilities of several sets of quantum channels which extend the set LOCC. We not only established a precise and accessible one-shot description of entanglement distillation under a wide variety of relevant operations, we revealed several operational equivalences in distillation in the one-shot regime --- showing in particular that all sets of free operations achieve exactly the same performance in pure-state distillation, with similar simplifications occurring also in the distillation from isotropic and maximally correlated states. The theoretical framework allowed us to establish computable expressions for the distillation fidelities and rates in such cases, thus providing an exact characterisation of entanglement distillation for these classes of states. The insight from the one-shot characterisation allowed for an operational interpretation of quantities which did not enjoy a direct interpretation of this kind, including asymptotic bounds such as the Rains bound and its zero-error equivalent as well as entanglement monotones such as the generalised robustness or the modified trace distance of entanglement.

Our work thus sheds light on fundamental problems in the study of manipulating entanglement as a resource. By providing a powerful theoretical framework, establishing a precise description of entanglement distillation in the practically relevant one-shot setting, as well as uncovering several novel relations in the operational description of LOCC and beyond, our results will contribute to the ongoing effort to efficiently utilise entanglement in technological applications and optimise the performance of quantum technologies.

Due to the high generality of our framework, we expect it to find use in a variety of contexts not explicitly considered in this work, facilitating the precise description of other classes of states and operations. We hope the results can aid not only the further study of entanglement, but also other quantum resources whose distillation enjoys a similar structure \cite{brandao_2015,liu_2019}, including for example coherence \cite{regula_2017,zhao_2019,lami_2019-1} or thermodynamics \cite{aberg_2013,horodecki_2013,yungerhalpern_2016}.


\begin{acknowledgments}

We thank Ryuji Takagi, Hayata Yamasaki, and in particular Ludovico Lami for useful discussions. B.R. and M.G. acknowledge the financial support of the National Research Foundation of Singapore Fellowship No. NRF-NRFF2016-02 and the National Research Foundation and L'Agence Nationale de la Recherche joint Project No. NRF2017-NRFANR004 VanQuTe.
K.F. was supported by the University of Cambridge Isaac Newton Trust Early Career grant RG74916. X.W. was supported by the Department of Defense.
\end{acknowledgments}

\bibliographystyle{apsrev4-1}
\bibliography{main}


\appendix

\section{Properties of the monotones $\TT{m}$}

First of all, we establish that the considered quantities are valid measures of entanglement. A common set of requirements that an entanglement monotone $M$ should obey is~\cite{vidal_2000}: faithfulness (i.e. $M(\rho) = 0$ iff $\rho \in \SEP$), convexity, and strong monotonicity (i.e. the requirement that $M(\rho) \geq \sum_i p_i M\left(\Lambda_i(\rho)\right)$ for any probabilistic protocol which applies an LOCC transformation $\Lambda_i$ to $\rho$ with probability $p_i$). By a direct application of Thm. 20 in~\cite{regula_2018}, we have the following.
\begin{proposition}\label{prop:monotone}
Let $\S \in \{\PPT, \PPTP, \SEP\}$, and consider the class of CPTP operations $\O$ such that $X \in \S \Rightarrow \Lambda(X) \in \S$. Then, for each $m\geq 1$,  $\TT{m}$ is faithful with respect to the set $\S$, convex, and strongly monotonic under the operations $\O$.\end{proposition}
The above establishes in particular that all of the measures are strong monotones under LOCC. Note, however, that $\TTT{m}_{\PPT}$ and $\TTT{m}_{\PPTP}$ are not faithful as entanglement measures, since they are zero for all PPT states.

The following result establishes a dual form for the measures.
\begingroup
\renewcommand{\thetheorem}{\ref{prop:proj_pos_neg}}
\begin{proposition}The measures $\TTT{m}_{\S}$ can be equivalently expressed as
\begin{equation}\begin{aligned}\label{eq:Tm_primal}
  \TTT{m}_{\S} (\rho) = \min \lset m \Tr\left(\rho-X\right)_+ + \Tr\left(\rho-X\right)_- \bar X \in \S\*\* \rset,
\end{aligned}\end{equation}
where $(\rho-X)_+$ (respectively, $(\rho-X)_-$) denotes the positive (negative) part of the Hermitian operator $\rho-X$.
\end{proposition}
\begin{proof}
For any self-adjoint operator $X$, let $\{X\mgeq0\}$ (respectively, $\{X\mleq 0\}$) denote the orthogonal projection operator onto the span of the eigenvectors corresponding to non-negative (non-positive) eigenvalues of $X$. The positive and negative parts of $X$ are then given by $X_+ = \{X\mgeq0\} X \{X\mgeq0\}$ and $X_- = -\{X\mleq0\} X \{X\mleq0\}$, such that $X = X_+ - X_-$.

By strong Lagrange duality we have
\begin{equation}\begin{aligned}
  &\TTT{m}_{\S} (\rho) =\\ &\min \lset m \Tr A + \Tr B \bar \rho - X = A - B,\; A,B \mgeq 0,\; X \in \S\*\* \rset.
\end{aligned}\end{equation}
We will now show that for each feasible $X$, the optimal value of the optimisation problem
\begin{equation}\begin{aligned}
  \min \lset m \Tr A + \Tr B \bar \rho - X = A - B,\; A,B \mgeq 0 \rset
\end{aligned}\end{equation}
is given by $m \Tr\left(\rho-X\right)_+ + \Tr\left(\rho-X\right)_-$. To see this, note that on the one hand we can take $A = (\rho - X)_+$ and $B = (\rho - X)_-$, and on the other hand by strong Lagrange duality we have
\begin{equation}\begin{aligned}
&\min \lset m \Tr A + \Tr B \bar \rho - X = A - B,\; A,B \mgeq 0 \rset\\
=&\max \lset \<\rho - X, W\> \bar -\id \mleq W \mleq m\id \rset
\end{aligned}\end{equation}
for which $W = m\{\rho \mgeq X\} - \{\rho \mleq X\}$ is a feasible solution.
\end{proof}
\endgroup

We additionally establish an equality between the two $\PPT$-based monotones in the case $m=d-1$.
\begin{proposition}\label{prop:rob_ppt_equal}
For any state we have
\begin{equation}\begin{aligned}
  R^\DD_\ppt(\rho) = R^\DD_{\PPTP}(\rho).
\end{aligned}\end{equation}
\end{proposition}
\begin{proof}
Let $W \in -\PPTPd$ be the optimal dual solution for $R^\DD_{\pptp}$, which means it satisfies $-\id \mleq W$ and $W = N + Q^{T_B}$ for $N, Q \mleq 0$. But then notice that $W' \coloneqq W - N \in -\PPT\*$ is also feasible, and we have $\< \rho, W' \> = \< \rho, W \> - \< \rho, N \> \geq \< \rho, W \>$, so in fact it suffices to optimise over $W \in -\PPT\*$.
\end{proof}

\begin{remark}
In~\cite{brandao_2005}, it was claimed that $\TTT{m-1}_\PPTP (\rho_f) = m f - 1$ for all $1 \leq m\leq d-1$ and $f \geq \frac1d$. One can see that this is incorrect by comparing it with the result of Thm.~\ref{thm:isotropic} which shows that $\TTT{m-1}_\PPTP (\rho_f) = \frac{(m-1)(df - 1)}{d-1}$.
\end{remark}

\begin{remark}A formula was given for $R^\DD_\PPTP$ in~\cite{brandao_2005} as
\begin{equation}\begin{aligned}
  R^\DD_\PPTP(\rho) = N(\rho) / \lambda_{\max}\left(\{\rho^{T_B}\mleq 0\}^{T_B}\right),
\end{aligned}\end{equation}
with $N(\rho)$ denoting the negativity. However, numerical counterexamples to this result can be readily constructed, and we find that this only provides a lower bound for the value of $R^\DD_\PPTP$ in general. For completeness, we give an explicit counterexample.
\end{remark}
Consider the symmetric and antisymmetric subspaces in $d=d_A=d_B=3$, spanned by the sets of mutually orthogonal vectors
\begin{equation}\begin{aligned}
  & \begin{cases} \ket{\psi_1} = \frac{1}{\sqrt{2}}\left(\ket{01} + \ket{10}\right)\\\ket{\psi_2} = \frac{1}{\sqrt{2}}\left(\ket{02} + \ket{20}\right)\\\ket{\psi_3} = \frac{1}{\sqrt{2}}\left(\ket{12} + \ket{21}\right)\end{cases}
  & \begin{cases} \ket{\alpha_1} = \frac{1}{\sqrt{2}}\left(\ket{01} - \ket{10}\right)\\\ket{\alpha_2} = \frac{1}{\sqrt{2}}\left(\ket{02} - \ket{20}\right)\\\ket{\alpha_3} = \frac{1}{\sqrt{2}}\left(\ket{12} - \ket{21}\right) \end{cases}
\end{aligned}\end{equation}
respectively.

Take the ansatz $\rho = \frac12\proj{\psi_1} + \frac12\proj{\psi_2}$. An explicit calculation gives $N(\rho) =  \frac{1}{2\sqrt{2}}$ and $\lambda_{\max}\left(\{\rho^{T_B}\mleq 0\}^{T_B}\right) = \frac{1}{2}$. Now, take the operator given by
\begin{equation}\begin{aligned}
  W = &\proj{\psi_1} + \proj{\psi_2} - \proj{\psi_3}\\ -& \proj{\alpha_1} - \proj{\alpha_2} + \proj{\alpha_3}\\ -& \proj{00} - \proj{11} - \proj{22},
\end{aligned}\end{equation}
whose partial transpose can be computed as $W^{T_B} = -3 \proj{w}$ with $\ket{w} = \frac{1}{\sqrt{3}}\left(\ket{00} - \ket{11} - \ket{22}\right)$. We then clearly have $W\mgeq-\id$ and $W^{T_B} \mleq 0$, which gives
\begin{equation}\begin{aligned}
  R_\PPTP(\rho) \geq \< \rho, W \> = 1 > \frac{1}{\sqrt{2}} = \frac{N(\rho)}{\lambda_{\max}\left(\{\rho^{T_B}\mleq 0\}^{T_B}\right)}.
\end{aligned}\end{equation}

We remark that, despite the lack of an exact analytical expression, the quantity $R^\DD_\PPTP$ can nevertheless be evaluated efficiently as a semidefinite program.

\begin{remark}
For completeness, we collect the results obtained in the manuscript which simplify the computation of $\TT{m}$ in several cases:
\begin{itemize}[leftmargin=*]
\item For any pure state, any $\S \in \{\SEP,\PPT,\PPTP\}$, and any integer $1 \leq m \leq d-1$, it holds that 
\begin{equation}\begin{aligned}
  \TTT{m-1}_\S(\psi) = \mnorm{\ket{\xi_\psi}}^2 - 1.
\end{aligned}\end{equation}
(See Thms. \ref{thm:Tm_pure_states} and \ref{thm:all_fid_equal_LOCC}.)

  \item For any isotropic state, any $\S \in \{\SEP,\PPT,\PPTP\}$, and any $1 \leq m \leq d-1$, it holds that
  \begin{equation}\begin{aligned}
     \TTT{m}_\S (\rho_f) = \begin{cases} 0 & f \leq \frac1d,\\ \frac{m(df - 1)}{d-1} & f \geq \frac1d.\end{cases} 
  \end{aligned}\end{equation}
  (See Thm.~\ref{thm:isotropic}.)

  \item For any maximally correlated state, any $\S \in \{\SEP,\PPT,\PPTP\}$, and any $1 \leq m \leq d-1$, it holds that
  \begin{equation}\begin{aligned}
    \TTT{m}_\S (\rho_\mc) = \TTT{m}_\I (\widetilde{\rho_\mc})
  \end{aligned}\end{equation}
  where $\rho_\mc = \sum_{i,j} \rho_{ij} \ketbra{ii}{jj}$, $\widetilde{\rho_\mc} = \sum_{i,j} \rho_{ij} \ketbra{i}{j}$, and $\I$ is the set of states diagonal in the given basis $\{\ket{i}\}$. (See Thm.~\ref{thm:max_corr_ppt})

  \item For any $\S \in \{\SEP,\PPT,\PPTP\}$, it holds that
  \begin{equation}\begin{aligned}
    \TTT{d-1}_\S (\rho) = R^\DD_\S (\rho).
  \end{aligned}\end{equation}
  (See Prop.~\ref{prop:rob_for_d-1}.)
\end{itemize}
\end{remark} 
\end{document}